\documentclass[letterpaper, 11pt]{article}

\usepackage{geometry}
\usepackage[utf8]{inputenc}
\usepackage{xcolor}
\usepackage{amsmath,amsthm,amsfonts,amssymb}
\usepackage{tikz}
\usepackage{enumitem}
\usepackage{tabularx}
\usepackage[colorlinks, allcolors=blue]{hyperref}
\usepackage[sort, numbers, compress]{natbib}

\providecommand{\keywords}[1]
{
  \small	
  \textbf{\textbf{Keywords---}} #1
}

\title{Approximate Byzantine Fault-Tolerance in \\Distributed Optimization
\thanks{\textbf{This updated version contains an alternative analysis of the proposed algorithm with CGE, making better use of the $2f$-redundancy property in Appendix~\ref{appdx:conv-guarantee-2}}. This report serves as the full version of the accepted paper with the same title \cite{publishedver}.}
}
\author{Shuo Liu \thanks{Georgetown University. Email: {\tt sl1539@georgetown.edu}.} \hspace{0.5in} Nirupam Gupta \thanks{École Polytechnique Fédérale de Lausanne (EPFL). Email: {\tt nirupam.gupta@epfl.ch}.} \hspace{0.5in} Nitin H. Vaidya \thanks{Georgetown University. Email: {\tt nitin.vaidya@georgetown.edu}.}
}
\date{}

\newtheorem{assumption}{\bfseries Assumption}
\newtheorem{definition}{\bfseries Definition}
\newtheorem{theorem}{\bfseries Theorem}
\newtheorem{lemma}{\bfseries Lemma}

\newtheorem*{lemma*}{Lemma}

\providecommand{\iprod}[2]{\ensuremath{\left\langle #1,\,#2  \right\rangle}}
\providecommand{\norm}[1]{\ensuremath{\left\lVert#1\right\rVert }}
\providecommand{\mnorm}[1]{\ensuremath{\left\lvert#1\right\rvert}}
\providecommand{\dist}[2]{\ensuremath{\mathrm{dist}\left( #1,\,#2 \right)}}
\providecommand{\rank}[1]{\ensuremath{\mathrm{rank}\left( #1 \right)}}

\providecommand{\comment}[1]{}
\providecommand{\nitin}[1]{}

\def\R{\mathbb{R}}
\def\H{\mathcal{H}}
\def\B{\mathcal{B}}

\def\O{\mathcal{O}}

\def\D{\mathsf{D}}
\def\L{\mathsf{L}}
\def\M{\mathsf{M}}
\def\W{\mathcal{W}}

\def\E{\mathbb{E}}

\def\gf{\mathsf{GradFilter}}

\begin{document}

\maketitle

\begin{abstract}
    This paper considers the problem of Byzantine fault-tolerance in distributed multi-agent optimization. In this problem, each agent has a local cost function, and in the fault-free case, the goal is to design a distributed algorithm that allows all the agents to find a minimum point of all the agents' aggregate cost function. We consider a scenario where some agents might be Byzantine faulty that renders the original goal of computing a minimum point of all the agents' aggregate cost vacuous.
    A more reasonable objective for an algorithm in this scenario is to allow all the non-faulty agents to compute the minimum point of only the non-faulty agents' aggregate cost. Prior work~\cite{gupta2020fault} shows that if there are up to $f$ (out of $n$) Byzantine agents then a minimum point of the non-faulty agents' aggregate cost can be computed {\em exactly} {if and only if} the non-faulty agents' costs satisfy a certain redundancy property called \textit{$2f$-redundancy}. However, $2f$-redundancy is an ideal property that can be satisfied only in systems free from noise or uncertainties, which can make the goal of exact fault-tolerance \textit{unachievable} in some
    applications.
    Thus, we introduce the notion of $(f,\epsilon)$-resilience, a generalization of exact fault-tolerance
    wherein the objective is to find an approximate minimum point of the non-faulty aggregate cost, with $\epsilon$ accuracy. This approximate fault-tolerance can be achieved under a weaker condition that is easier to satisfy in practice, compared to $2f$-redundancy.
    We obtain necessary and sufficient conditions for achieving $(f, \, \epsilon)$-resilience characterizing the correlation between relaxation in redundancy and approximation in resilience. In case when the agents' cost functions are differentiable, we obtain conditions for $(f, \, \epsilon)$-resilience of the distributed gradient-descent method when equipped with \textit{robust gradient aggregation}; such as \textit{comparative gradient elimination} or \textit{coordinate-wise trimmed mean}. 
\end{abstract}





\newpage
\tableofcontents
\newpage
\comment{This version now incorporates all comments from Nitin's email on Aug 11.\\
TODO:\\
- New necessity proof\\
- Check numerical experiments and Fig. 2
}

\section{Introduction}

The problem of distributed optimization in multi-agent systems has gained significant attention in recent years~\cite{boyd2011distributed, nedic2009distributed, duchi2011dual}. In this problem, each agent has a {\em local cost function} and, when the agents are fault-free, the goal is to design algorithms that allow the agents to collectively minimize the aggregate of their cost functions. To be precise, suppose that there are $n$ agents in the system and let $Q_i(x)$ denote the {local} cost function of agent $i$, where $x$ is a $d$-dimensional vector of real values, i.e., $x\in \R^d$. A traditional distributed optimization algorithm outputs a {\em global minimum} $x^*$ such that
\begin{align}
    x^* \in \arg \min_{x \in \R^d} ~ \sum_{i = 1}^n Q_i(x). \label{eqn:obj}
\end{align}
As a simple example, $Q_i(x)$ may denote the cost for an agent $i$ (which may be a robot or a person) to travel to location $x$ from their current location,
and $x^*$ is a location that minimizes the total cost of meeting for all the agents. Such multi-agent optimization 
is of interest in many practical applications,
 including distributed machine learning~\cite{boyd2011distributed}, swarm robotics~\cite{raffard2004distributed}, and distributed sensing~\cite{rabbat2004distributed}.\\

We consider the distributed optimization problem in the presence of up to $f$ Byzantine faulty agents, originally introduced by Su and Vaidya \cite{su2016fault}. The Byzantine faulty agents may behave arbitrarily~\cite{lamport1982byzantine}. In particular, the non-faulty agents
may share arbitrary incorrect and inconsistent information in order to bias the output of a distributed optimization algorithm. For example, consider an application
of multi-agent optimization in the case of distributed sensing where the agents (or {\em sensors}) observe a common {\em object} in order to collectively identify the object. However, the faulty agents may send arbitrary observations concocted to prevent the non-faulty agents from making the correct identification~\cite{chen2018resilient, chong2015observability, pajic2014robustness, su2016non}. Similarly, in the case of distributed learning, which is another application of distributed optimization, the faulty agents may send incorrect information based on {\em mislabelled} or arbitrary concocted data points to prevent the non-faulty agents from learning a {\em good} classifier~\cite{alistarh2018byzantine, bernstein2018signsgd, blanchard2017machine, cao2019distributed, charikar2017learning, chen2017distributed, gupta2019byzantine_allerton, xie2018generalized}. 

\subsection{Background: Exact Fault-Tolerance}
\label{sub:background}

In the {\em exact fault-tolerance} problem, the goal is to design a distributed algorithm that allows all the non-faulty agents to compute a minimum point of the aggregate cost of only the non-faulty agents~\cite{gupta2020fault}. Specifically, suppose that in a given
execution, set $\B$ with $\mnorm{\B} \leq f$ is the set of 
Byzantine agents, where notation $\mnorm{\cdot}$ denotes the set cardinality, and $\H = \{1, \ldots, \, n\} \setminus \B$ denotes the set of non-faulty (i.e., honest) agents. Then, a distributed optimization algorithm has exact fault-tolerance if it outputs a point $x^*_{\H}$ such that
\begin{align}
    x^*_{\H} \in \arg \min_{x \in \R^d} ~ \sum_{i \in \H} Q_i(x). \label{eqn:honest_obj}
\end{align}
However, since the identity of the Byzantine agents is a priori unknown, in general, exact fault-tolerance is unachievable \cite{su2016fault}. Specifically, as shown in~\cite{gupta2020fault, gupta2020resilience}, exact fault-tolerance can be achieved {\em if and only if} the agents' cost functions satisfy the {\em $2f$-redundancy} property defined below. 


\begin{definition}[{\bf $2f$-redundancy}] 
\label{def:red}
The agents' cost functions are said to have {\em $2f$-redundancy} property if and only if  
for every pair of subsets $S, \,\widehat{S} \subseteq \{1, \ldots, \, n\}$ with $\widehat{S} \subseteq S$, $\mnorm{S} = n-f$, and $\mnorm{\widehat{S}} \geq n-2f$,
$$\arg \min_{x \in \R^d} ~ \sum_{i \in \widehat{S}}Q_i(x) = \arg \min_{x \in \R^d} ~ \sum_{i \in S} Q_i(x).$$
\end{definition}

In principle, the {$2f$-redundancy} property can be realized by design for many applications of multi-agent distributed optimization including distributed sensing and distributed learning (see~\cite{gupta2019byzantine, gupta2020fault}). However, practical realization of {$2f$-redundancy} can be difficult in the presence of {\em noise} in the real-world systems. Therefore, we propose a pragmatic generalization of exact fault-tolerance, namely {\em $(f,\epsilon)$-resilience}. 

\subsection{$(f,\epsilon)$-Resilience: A Relaxation of Exact Fault-Tolerance}
\label{sub:approx}

Intuitively, the proposed notion of {\em $(f,\epsilon)$-resilience} requires an algorithm to output {\em approximation} of a minimum point of the aggregate of the cost functions of sufficiently large subsets of non-faulty agents. We define {\em $(f, \, \epsilon)$-resilience} below, where $\epsilon \in \R_{\geq 0}$ is the measure of approximation and $\norm{\cdot}$ denotes the Euclidean norm. The Euclidean distance between a point $x$ and a non-empty set $X$ in space $\R^d$ is denoted by \dist{x}{X}, and is defined as
\begin{align}
    \dist{x}{X} = \inf_{y\in X}\norm{x-y}. \label{eqn:dist_pt_set}
\end{align}

\begin{definition}[\textbf{$(f, \, \epsilon)$-resilience}]
\label{def:approx_res}
A distributed optimization algorithm is said to be $(f, \, \epsilon)$-resilient if it outputs a point $\widehat{x} \in \R^d$ such that for every subset $S$ of non-faulty agents with $\mnorm{S} = n-f$, 
\[\dist{\widehat{x}}{\arg \min_{x \in \R^d} \sum_{i \in S}Q_i(x)} \leq \epsilon,\]
despite the presence of up to $f$ Byzantine agents.
\end{definition}

Thus, with {\em $(f, \, \epsilon)$-resilience}, the output is within distance $\epsilon$ of a minimum point of the aggregate cost function of any $n-f$ non-faulty agents. As there can be at most $f$ Byzantine faulty agents whose identity remains unknown, the following two scenarios are indistinguishable in general: (1) there are exactly $f$ Byzantine agents, and (2) there are less than $f$ Byzantine agents. Thus, estimation for the minimum point of the aggregate cost functions of $n-f$ non-faulty agents is indeed a reasonable goal~\cite{su2016fault}. 
Analogous resilience requirements have been previously studied in other contexts as well, such as robust statistics (e.g., robust mean estimation~\cite{charikar2017learning,steinhardt2017resilience}) and fault-tolerant linear state estimation~\cite{bhatia2015robust,fawzi2014secure, feng2014distributed,mishra2016secure, pajic2017attack, shoukry2017secure}. In this work, we address resilience in the context of distributed optimization.\\

In this paper, we only consider {\em deterministic} algorithms which, given a fixed set of inputs from the agents, always output the same point in $\R^d$. Thus, a deterministic $(f, \, \epsilon)$-resilient algorithm produces a unique output point in all of its executions with identical inputs from all the agents (including the faulty ones). 
{\bf Note that} in the deterministic framework, exact fault-tolerance is equivalent to {$(f, \, 0)$-resilience}, i.e., a deterministic {$(f, \, 0)$-resilient} algorithm achieves exact fault-tolerance, and vice-versa.\footnote{Refer to Appendix~\ref{app:exact_approx} for proof.} 
Therefore, results on $(f, \, \epsilon)$-resilience for arbitrary $\epsilon \geq 0$ have a wider application compared to results applicable only to exact fault-tolerance, e.g.,~\cite{blanchard2017machine, guerraoui2018hidden, gupta2020fault, su2018finite}. \\

We show that {\em $(f, \, \epsilon)$-resilience} requires a {\em weaker redundancy} condition, in comparison to $2f$-redundancy, named {\em $(2f, \, \epsilon)$-redundancy} defined in Definition \ref{def:approx_red} below. Recall that the {\em Euclidean Hausdorff distance} between two sets $X$ and $Y$ in $\R^d$, which we denote by $\dist{X}{Y}$, is defined as follows~\cite{munkres2000topology}:
\begin{equation}
    \dist{X}{Y}\triangleq\max\left\{\sup_{x\in X} \dist{x}{Y}, ~ \sup_{y\in Y}\dist{y}{X}\right\}. \label{eqn:hausdorff}
\end{equation}

\begin{definition}[\textbf{$(2f, \, \epsilon)$-redundancy}]
\label{def:approx_red}
The agents' cost functions are said to have {\em $(2f, \, \epsilon)$-redundancy} property if and only if for every pair of subsets $S, \,\widehat{S} \subseteq \{1, \ldots, \, n\}$ with
$\mnorm{S} = n-f$, $\mnorm{\widehat{S}} = n-2f$ and $\widehat{S} \subseteq S$,
\begin{equation}
    \dist{\arg\min_{x \in \R^d}\sum_{i\in S}Q_i(x)}{\arg\min_{x \in \R^d}\sum_{i\in \widehat{S}}Q_i(x)} \leq \epsilon. \label{eqn:red_dist}
\end{equation}
\end{definition}

It is easy to show that $2f$-redundancy (Definition~\ref{def:red}) is equivalent to $(2f, \, 0)$-redundancy (note that $\epsilon=0$ here). It is also obvious that $2f$-redundancy implies $(2f, \, \epsilon)$-redundancy for all $\epsilon \geq 0$. However, the converse need not be true. Thus, the $(2f, \, \epsilon)$-redundancy property with $\epsilon > 0$ is {\em weaker} than $2f$-redundancy. 

\subsection{Applications}
\label{sub:app}
Our results are applicable to a large class of distributed optimization problems; including distributed sensing~\cite{chong2015observability, pajic2017attack, pajic2014robustness, su2018finite},  distributed machine learning~\cite{bottou2018optimization, boyd2011distributed, chen2017distributed, yin2018byzantine}, and distributed linear regression (Section \ref{sec:experiments}).
We discuss below the specific case of distributed learning.\\

{\bf Distributed Learning:}
In this particular optimization problem, each agent has some local {\em data points} and the goal for the agents is to compute a learning parameter that best models the collective data points observed by all the agents~\cite{bottou2018optimization}. Specifically, given a learning parameter $x$, for each data point $z$, we define a loss function $\ell(x; \, z)$. Suppose that the data generating distribution of agent $i$ is $\mathcal{D}_i$,
and let $\E_{z \sim \mathcal{D}_i}$ denote the expectation with respect to the random data point $z$ over distribution $\mathcal{D}_i$.
Then, 
\[Q_i(x) \triangleq \E_{z \sim \mathcal{D}_i} ~ \ell (x; z)\]

When the distribution of data points is identical for all the agents then the $2f$-redundancy property holds true. However, in practice this is rarely the case~\cite{chen2017distributed, yin2018byzantine}. Indeed, different agents may have different data distributions in practice. Therefore, exact fault-tolerance in a pragmatic distributed learning framework is an extremely difficult (if not impossible) goal. 
In the context of distributed learning, our results on approximate fault-tolerance characterize the relationship between the correlation amongst different agents' data (i.e., degree of redundancy), and the fault-tolerance achieved.

\subsection{System architecture}

We consider {\em synchronous} systems.
Our results apply to the two architectures shown in Figure \ref{fig:sys}. In the server-based architecture, the server is assumed to be trustworthy, but up to $f$ agents may be Byzantine faulty.
In the peer-to-peer architecture, the agents are connected by a complete network, and up to $f$ of these agents may be Byzantine faulty. Provided that $f<\frac{n}{3}$, an algorithm for the server-based architecture can be simulated in the peer-to-peer system using
the well-known {\em Byzantine broadcast} primitive~\cite{lynch1996distributed}. For simplicity of presentation, the rest of this paper considers the server-based architecture. 

\begin{figure}[htb!]
\centering
\includegraphics[width = 0.5\linewidth]{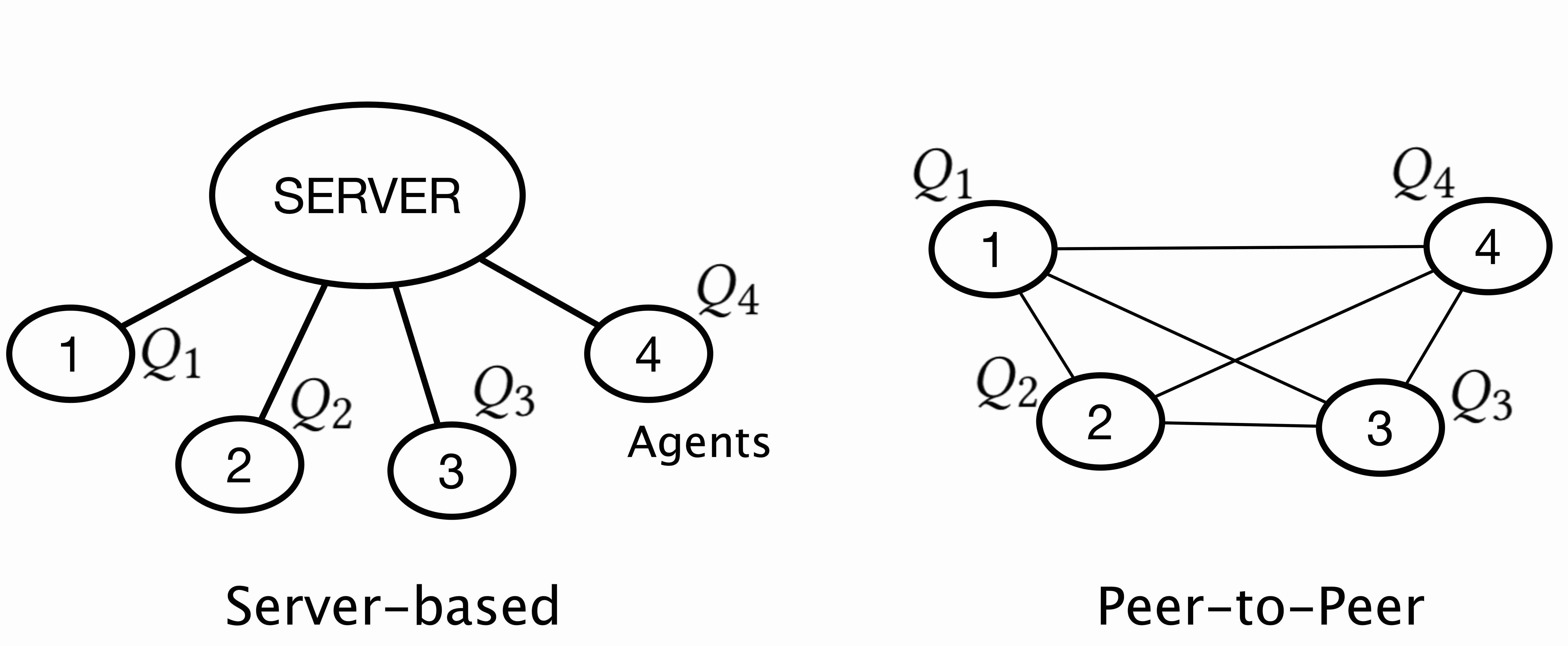}
\caption{\small{\it System architecture.}}
\label{fig:sys}
\end{figure}

\subsection{Summary of Our Contributions}
\label{sub:contri}

In the {\bf first part} of the paper, i.e., Section~\ref{sec:approx-fault-tolerance}, we obtain conditions on feasibility and achievability of approximate fault-tolerance of desirable accuracy. Specifically, we show that 

\begin{itemize}
\setlength{\itemsep}{0.3em}
    \item $(f, \epsilon)$-resilience is feasible only if $(2f, \epsilon)$-redundancy property holds true.
    \item If $(2f, \epsilon)$-redundancy property holds true then $(f, 2\epsilon)$-resilience is achievable.
\end{itemize}

In the {\bf second part}, i.e., Sections~\ref{sec:grad_des} and~\ref{sec:experiments}, we consider the case when agents' costs are differentiable, such as in machine learning~\cite{bottou2018optimization, yang2017byrdie}, or regression~\cite{gupta2019byzantine, su2018finite, su2016non}. We consider the distributed gradient-descent (DGD) method - an iterative distributed optimization algorithm commonly used in this particular case.

\begin{itemize}
\setlength{\itemsep}{0.3em}
    \item We propose a generic sufficient condition for convergence of the DGD method equipped with a {\em gradient-filter} (also referred as {\em robust gradient aggregation}), which is a common fault-tolerance mechanism, e.g., see~\cite{blanchard2017machine, chen2017distributed, gupta2020fault, yin2018byzantine}.
    
    \item Later, in Section~\ref{sec:grad_filters}, we utilize the above result to obtain approximate fault-tolerance properties of the following two specific gradient-filters, under $(2f, \, \epsilon)$-redundancy: (i) Comparative gradient elimination (CGE)~\cite{gupta2019byzantine}, and (ii) Coordinate-wise trimmed mean (CWTM)~\cite{su2018finite}.
    These two gradient-filters are both easy to implement and versatile~\cite{gupta2020byzantine, gupta2020fault, su2018finite, yin2018byzantine}. 
    
    \item Finally, in Section~\ref{sec:experiments}, we present empirical comparisons between approximate fault-tolerance of the two gradient-filters by simulating a problem of distributed linear regression.
\end{itemize}
~

{\bf Note:} As $(f, \, 0)$-resilience is equivalent to exact fault-tolerance (see Section~\ref{sub:approx}), our results on $(f, \, \epsilon)$-resilience encapsulate all the existing results applicable only to exact fault-tolerance, such as the ones in~\cite{blanchard2017machine, guerraoui2018hidden, gupta2020fault, su2018finite}.  
~\\

Compared to related works~\cite{karimireddy2020learning, kuwaranancharoen2020byzantine}, we present precise redundancy conditions needed for obtaining Byzantine fault-tolerance within a specified approximation error. Unlike them, our results on the impossibility and feasibility of approximate fault-tolerance are applicable to non-differentiable cost functions. Moreover, in the case when the cost functions are differentiable, we present a generic condition for convergence of the DGD method that can precisely model the approximate fault-tolerance property of a generic robust gradient-aggregation rule (a.k.a., gradient-filter). \\

This is the full version of the paper including proofs of theorems and additional experimental results and discussion.

\section{Other Related Work}
\label{sec:prior}

In the past, different notions of approximate fault-tolerance, besides $(f, \, \epsilon)$-resilience, have been used to analyze Byzantine fault-tolerance of different distributed optimization algorithms~\cite{diakonikolas2018sever, su2016fault}. As we discuss below in Section~\ref{sub:alt_approx}, the difference between these other definitions and our definition of $(f, \, \epsilon)$-resilience arises mainly due to the applicability of the distributed optimization problems. Later, in Section~\ref{sub:gf_prior}, we discuss some prior work on gradient-filters used for achieving Byzantine fault-tolerance in the distributed gradient-descent method.

\subsection{Alternate Notions of Approximation in Fault-Tolerance}
\label{sub:alt_approx}

As proposed by Su and Vaidya, 2016~\cite{su2016fault}, 
instead of a minimum point of the {\em uniformly weighted} aggregate of non-faulty agents' cost functions, a distributed optimization algorithm may output a minimum point of a {\em non-uniformly weighted} aggregate of non-faulty costs, i.e., $\sum_{i \in \H} \alpha_i \, Q_i(x)$,
where $\H$ denotes the set of at least $n-f$ non-faulty agents, and $\alpha_i \geq 0$ for all $i \in \H$. As is suggested in~\cite{su2016fault}, upon re-scaling the coefficients such that $\sum_{i \in \H} \alpha_i = 1$, we can measure approximation in fault-tolerance using two metrics: (1) the number of coefficients in $\{\alpha_i, \, i \in \H\}$ that are positive, and (2) the minimum positive value amongst the coefficients: $\min \left\{ \alpha_i; \, \alpha_i > 0, \, i \in \H \right\}$. Results on the achievability of this particular form of approximation for the scalar case (i.e., $d = 1$) can be found in~\cite{su2016fault, su2020byzantine}. However, we are unaware of similar results for the case of higher-dimensional optimization problem, i.e., when $d > 1$. There is some work on this particular notion of approximate fault-tolerance in high-dimensions, such as~\cite{su2018finite, yang2017byrdie}, however their results only apply to special cost functions, specifically, quadratic or strictly convex functions, as opposed to the generic cost functions (that need not even be differentiable) considered in this paper. \\

Another way of measuring approximation is by the value of the aggregate cost function, or its gradient. For instance, as discussed in~\cite{diakonikolas2018sever}, for the case of differentiable cost functions a resilient distributed optimization algorithm $\Pi$ may output a point $x_{\Pi} \in \R^d$ such that each element of the aggregate non-faulty gradient $\sum_{i \in \H} \nabla Q_i(x_{\Pi})$ is bounded by $\epsilon$. As yet another alternative, a resilient algorithm $\Pi$ may aim to output a point $x_\Pi$
such that the non-faulty aggregate cost $\sum_{i \in \H} Q_i(x_\Pi)$ is within $\epsilon$ of the true minimum cost $\min_x \sum_{i \in \H} Q_i(x)$. However, these definitions of approximate resilience are sensitive to scaling of the cost functions. In particular, if the elements of $\sum_{i \in \H} \nabla Q_i(x_\Pi)$ are bounded by $\epsilon$ then the elements of $\sum_{i \in \H} \alpha \nabla Q_i(x_\Pi)$ are bounded by $\alpha \epsilon$, where $\alpha$ is a positive scalar value. On the other hand, both $\sum_{i \in \H} Q_i(x)$ and $\sum_{i \in \H} \alpha Q_i(x)$ have identical minimum point regardless of the value of $\alpha$. Therefore, when the objective is to approximate a minimum point of the non-faulty aggregate cost $\arg \min_x \sum_{i \in \H} Q_i(x)$, which is indeed the case in this paper, use of function (or gradient) values to measure approximation is not a suitable choice.

\subsection{Gradient-Filters}
\label{sub:gf_prior}

In the past, several {gradient-filters} have been proposed to {\em robustify} the distributed gradient-descent (DGD) method against Byzantine faulty agents in a server-based architecture, e.g., see~\cite{alistarh2018byzantine, blanchard2017machine, damaskinos2018asynchronous, diakonikolas2018sever, gupta2020byzantine, prasad2018robust, su2016fault, yin2018byzantine}. 
A gradient-filter refers to {\em Byzantine robust aggregation} of agents' gradients that mitigates the detrimental impact of incorrect gradients sent by the Byzantine agents to the server. To name a few gradient-filters, that are provably effective against Byzantine agents, we have the comparative gradient elimination (CGE)~\cite{gupta2019byzantine, gupta2020byzantine}, coordinate-wise trimmed mean (CWTM)~\cite{su2016fault, yin2018byzantine}, geometric median-of-means (GMoM)~\cite{chen2017distributed}, KRUM~\cite{blanchard2017machine}, Bulyan~\cite{guerraoui2018hidden}, and other spectral gradient-filters~\cite{diakonikolas2018sever}. Different gradient-filters guarantee some fault-tolerance under different assumptions on non-faulty agents' cost functions. \\

In this paper, we propose a generic result, in Theorem~\ref{thm:upper-bound-D} in Section~\ref{sec:grad_des}, on the convergence of the DGD method equipped with a gradient-filter. The result holds true regardless of the gradient-filter used, and thus, can be utilized to obtain formal fault-tolerance property of a gradient-filter in context of the considered distributed optimization problem. We demonstrate this, in Section~\ref{sec:grad_filters}, by obtaining $(f, \, \epsilon)$-resilience properties of two specific gradient-filters; CGE and CWTM. As exact fault-tolerance is equivalent to $(f, \, 0)$-resilience (see Section~\ref{sub:approx}), our results generalize the prior work on exact fault-tolerance of these two filters, see~\cite{gupta2019byzantine, gupta2020byzantine, su2018finite}. Moreover, until now, exact fault-tolerance of the CWTM gradient-filter was only studied for special optimization problems of state estimation~\cite{su2018finite}, and machine learning~\cite{yin2018byzantine}. Our result presents the fault-tolerance property of CWTM for a much larger class of optimization problems.

\subsection{Robust Statistics with Arbitrary Outliers}
As noted earlier, there has been work on the problem of robust statistics with arbitrary outliers~\cite{charikar2017learning,feng2014distributed,steinhardt2017resilience}. In this problem, we are given a finite set of data points; $\alpha$ fraction of which are sampled independently and identically from a common distribution $\mathcal{D}$ in $\R^d$, and the remaining $1 - \alpha$ fraction of data points may be arbitrary. The identity of arbitrary data points is a priori unknown, otherwise the problem is trivialized. The objective in this problem is to estimate statistical measures of distribution $\mathcal{D}$, such as mean, or variance, despite the presence of arbitrary outliers. The problem robust mean estimation can potentially be modelled as a fault-tolerant distributed optimization problem where for each non-faulty agent $i$, $Q_i(x): (x, \, x_i) \mapsto y \in \R$ for all $x \in \R^d$ where $x_i \sim \mathcal{D}$. A faulty agent may choose an arbitrary cost function. The cost functions can be designed in a manner such that the minimum point of the aggregate of non-faulty cost functions is equal to the mean for the non-faulty data points. 
%
In particular, suppose that for each non-faulty agent $i$, $Q_i(x) \triangleq \norm{x - x_i}^2$ where $x_i \sim \mathcal{D}$.
In this case, the minimum point of the non-faulty aggregate cost function is equal to the average of the non-faulty data points sampled from distribution $\mathcal{D}$.\\

Prior work on robust statistics considers a centralized setting wherein, unlike a distributed setting, all the data points are accessible to a single machine. In this report, we also present distributed algorithms that do not require the agents to share their local data points. Moreover, in the centralized setting, our results are applicable to a larger class of cost functions, including non-convex functions.

\subsection{\bf Fault-tolerance in State Estimation} 
The problem of distributed optimization finds direct application in distributed state estimation~\cite{rabbat2004distributed}. In this problem, the system comprises multiple sensors, and each sensor makes partial observations about the system's state. The goal is to compute the entire state of the system using collective observations from all the sensors. However, if a sensor is faulty then it may share incorrect observations, preventing correct state estimation. The special case of distributed state estimation when the observations are {\em linear} in the system's state has gained significant attention in the past, e.g. see~\cite{bhatia2015robust, chong2015observability, mishra2016secure, pajic2017attack, pajic2014robustness, shoukry2015imhotep, shoukry2017secure, su2018finite}. These works have shown that the state can be determined despite up to $f$ (out of $n$) faulty observations {\em if and only if} the system is {\em $2f$-sparse observable}, i.e., the complete state can be determined using observations of only $n-2f$ non-faulty sensors. We note that, in this particular case, {\em $2f$-sparse observability} is equivalent to $2f$-redundancy. Additionally, some of these works, such as~\cite{mishra2016secure, su2018finite}, also consider the case of {\em approximate} linear state estimation when the observations are noisy. Our work is more general in that we consider the problem setting of distributed optimization, and our results apply to a larger class of cost functions.


\section{Necessary and Sufficient Conditions for $(f, \, \epsilon)$-Resilience}
\label{sec:approx-fault-tolerance}
Throughout this paper we assume, as stated below, that the non-faulty agents' cost functions and their aggregates have well-defined minimum points. Otherwise, the problem of optimization is rendered vacuous. 
\begin{assumption}
\label{asp:exists}
For every non-empty set of non-faulty agents $S$, we assume that the set \\$\arg \min_{x \in \R^d} \sum_{i \in S}Q_i(x)$ is non-empty and closed.
\end{assumption}

%




We also assume that $f < n/2$. Lemma~\ref{lem:n_2f} below shows that $(f, \, \epsilon)$-resilience is impossible in general when $f \geq n/2$. Proof of Lemma~\ref{lem:n_2f} is easy, and can be found in Appendix~\ref{app:n_2f}.

\begin{lemma}
\label{lem:n_2f}
If $f \geq n/2$ then there cannot exist a deterministic $(f, \, \epsilon)$-resilient algorithm for any $\epsilon \geq 0$.
\end{lemma}

\subsection{Necessary Condition}
\comment{Email on Aug 11: Update this necessity proof}

\begin{theorem}
\label{thm:necc}
Suppose that Assumption~\ref{asp:exists} holds true. There exists a deterministic $(f, \,\epsilon)$-resilient distributed optimization algorithm where $\epsilon \geq 0$ {\em only if} the agents' cost functions satisfy the $(2f,\epsilon)$-redundancy property.
\end{theorem}

\begin{proof}

To prove the theorem we present a scenario when the agents' cost functions (if non-faulty) are scalar functions, i.e., $d = 1$ and for all $i$, $Q_i: \R \to \R$, and the minimum point of an aggregate of one or more agents' cost functions is uniquely defined. Obviously, if a condition is necessary in this particular scenario then it is so in the general case involving vector functions with non-unique minimum points.\\


To prove the necessary condition, we also assume that the server has full knowledge of all the agents' cost functions. This may not hold true in practice, where instead the server may only have partial information about the agents' cost functions. Indeed, this assumption forces the Byzantine faulty agents to a priori fix their cost functions. However, in reality the Byzantine agents may send arbitrary information over time to the server that need not be consistent with a fixed cost function. Thus, necessity of $(2f,\epsilon)$-redundancy under this strong assumption implies its necessity in general. \\

The proof is by contradiction. Specifically, we show that {\em If the cost functions of non-faulty agents do not satisfy the $(2f,\epsilon)$-redundancy property then there cannot exist a {\em deterministic} $(f,\epsilon)$-resilient distributed optimization algorithm.}\\
%
%

Recall that we have assumed that for a non-empty set of agents $T$ the aggregate cost function $\sum_{i \in T} Q_i(x)$ has a unique minimum point. To be precise, for each non-empty subset of agents $T$, we define
\[x_T = \arg \min_x \sum_{i \in T} Q_i(x).\]

Suppose that the agents' cost functions {\bf do not} satisfy the $(2f,\epsilon)$-redundancy property stated in Definition~\ref{def:approx_red}. 
Then, there exists a real number $\delta > 0$ and a pair of subsets $S, \, \widehat{S}$ with $\widehat{S} \subset S$, $\mnorm{S} = n-f$, and $n-2f \leq \mnorm{\widehat{S}} < n-f$ such that
\begin{align}
    \norm{x_{\widehat{S}} - x_S} \geq \epsilon + \delta. \label{eqn:no_red_asp}
\end{align}
Now, suppose that $n - f - \mnorm{\widehat{S}}$ agents in the remainder set $\{1, \ldots, \, n\} \setminus S$ are Byzantine faulty. Let us denote the set of faulty agents by $\B$. Note that $\B$ is non-empty with $\mnorm{\B} = n - f - \mnorm{\widehat{S}} \leq f$. 
Similar to the non-faulty agents, the faulty agents send to the server cost functions that are scalar, and the aggregate of one or more agents' cost functions in the set $S \cup \B$ is unique. However, 
the aggregate cost function of the agents in the set $\B \cup \widehat{S}$ minimizes at a unique point $x_{\B \cup \widehat{S}}$ which is $\norm{x_{\widehat{S}} - x_{S}}$ distance away from $x_{\widehat{S}}$, similar to $x_S$, but lies on the other side of $x_{\widehat{S}}$ as shown in the figure below. Note that it is always possible to pick such functions for the faulty agents.


\begin{center}
    \includegraphics{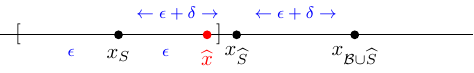}
\end{center}

\noindent Note that the distance between the two points $x_S$ and $x_{\B \cup \widehat{S}}$ is $2\epsilon + 2\delta$, i.e.,
\begin{align}
    \norm{x_S - x_{\B \cup \widehat{S}}} = 2\epsilon + 2\delta. \label{eqn:dist_x1_x2}
\end{align}

Now, suppose, toward a contradiction, that there exists an $(f, \, \epsilon)$-resilient deterministic optimization algorithm named $\Pi$. As the identity of Byzantine faulty agents is a priori unknown to the server, and the cost functions sent by the Byzantine faulty agents have similar properties as the non-faulty agents, 
the server cannot distinguish between the following two possible scenarios; i) $S$ is the set of non-faulty agents, and ii) $\B \cup \widehat{S}$ is the set of non-faulty agents. Note that both the sets $S$ and $\B \cup \widehat{S}$ contain $n-f$ agents. \\

As the cost functions received by the server are identical in both of the above scenarios, being a deterministic algorithm, $\Pi$ should have identical output in both the cases. We let $\widehat{x}$ denote the output of $\Pi$. In scenario (i) when the set of honest agents is given by $S$ with $\mnorm{S} = n-f$, as $\Pi$ is assumed $(f,\epsilon)$-resilient, by Definition~\ref{def:approx_res} the output
\begin{align}
    \widehat{x} \in  [x_S - \epsilon, \, x_S + \epsilon] \label{eqn:scn_1}
\end{align}
as shown in the figure above. Similarly, in scenario (ii) when the set of honest agents is $\B \cup \widehat{S}$ with $\mnorm{\B \cup \widehat{S}} = n-f$,
\begin{align}
    \widehat{x} \in  [x_{\B\cup \widehat{S}} - \epsilon, \, x_{\B\cup \widehat{S}} + \epsilon]. \label{eqn:scn_2}
\end{align}
However,~\eqref{eqn:dist_x1_x2} implies that~\eqref{eqn:scn_1} and~\eqref{eqn:scn_2} cannot be satisfied simultaneously. That is, if $\Pi$ is $(f,\epsilon)$-resilient in scenario (i) then it cannot be so in scenario (ii), and vice-versa. This contradicts the assumption that $\Pi$ is $(f, \, \epsilon)$-resilient. \qedhere
\end{proof}



\subsection{Sufficient Condition}



\begin{theorem}
\label{thm:suff}
Suppose that Assumption~\ref{asp:exists} holds true. 
For a real value $\epsilon \geq 0$, if the agents' cost functions satisfy the $(2f,\epsilon)$-redundancy property then $(f,2\epsilon)$-resilience is achievable.
\end{theorem}

\begin{proof}

The proof is constructive where we assume that all the agents send their individual cost functions to the server. We assume that $f > 0$ to avoid the trivial case of $f = 0$. Throughout the proof we write the notation $\arg \min_{x \in \R^d}$ simply as $\arg \min$, unless otherwise stated.
We begin by presenting an algorithm below, comprising three steps.


\begin{description}
    \item[Step 1:] Each agent sends their cost function to the server. An honest agent sends its actual cost function, while a faulty agent may send an arbitrary function.
    \item[Step 2:] For each set $T$ of received functions, $\mnorm{T}=n-f$, the server computes a point
        \begin{equation*}
            x_T\in\arg\min \sum_{i\in T}Q_i(x).
        \end{equation*}
        For each subset $\widehat{T}\subset T$, $\mnorm{\widehat{T}}=n-2f$, the server computes 
        \begin{equation}
            r_{T\widehat{T}}\triangleq\dist{x_T}{\arg\min \sum_{i\in\widehat{T}}Q_i(x)}, \label{eqn:r_SS}
        \end{equation}
        and 
        \begin{equation}
            r_T=\max_{\begin{subarray}{c}\widehat{T}\subset T,\\\mnorm{\widehat{T}}=n-2f\end{subarray}}r_{T\widehat{T}.} \label{eqn:r_S}
        \end{equation}
    \item[Step 3:] The server outputs $x_S$ such that 
    \begin{align}
        S = \underset{\begin{subarray}{c}T \subset \{1, \ldots, \, n\},\\ \mnorm{T} = n-f\end{subarray}}{\arg \min} r_T. \label{eqn:out_S}
    \end{align}
\end{description}

\noindent We show that above algorithm is $(f, \, 2\epsilon)$-resilient under $(2f, \, \epsilon)$-redundancy. For a non-empty set of agents $T$, we denote 
\[X_T = \arg\min \sum_{i\in T}Q_i(x).\]
~

Consider an arbitrary set of non-faulty agents $G$ with $\mnorm{G} = n-f$. Such a set is guaranteed to exist as there are at most $f$ faulty agents, and therefore, at least $n-f$ non-faulty agents exist in the system.
Consider an arbitrary set $\widehat{T}$ such that $\widehat{T}\subset G$ and $\mnorm{\widehat{T}}=n-2f$. By Definition~\ref{def:approx_red} of $(2f, \, \epsilon)$-redundancy, 
\begin{align}
    \dist{X_G}{X_{\widehat{T}}}\leq\epsilon. \label{eqn:G_T_eps}
\end{align}
Recall from~\eqref{eqn:r_SS} that $r_{G \widehat{T}} = \dist{x_G}{X_{\widehat{T}}}$. As $x_G \in X_G$, by Definition~\eqref{eqn:hausdorff} of Hausdorff set distance, $\dist{x_G}{X_{\widehat{T}}} \leq \dist{X_G}{X_{\widehat{T}}}$. Therefore, $r_{G \widehat{T}} \leq \dist{X_G}{X_{\widehat{T}}}$, and
substituting from~\eqref{eqn:G_T_eps} implies that
\begin{align}
    r_{G \widehat{T}} \leq \epsilon.  \label{eqn:r_G_T}
\end{align}
Now, recall from~\eqref{eqn:r_S} that $r_G = \max\left\{r_{G\widehat{T}} ~ \vline ~ \widehat{T}\subset G, \, \mnorm{\widehat{T}}=n-2f \right\}$.
As $\widehat{T}$ in~\eqref{eqn:r_G_T} is an arbitrary subset of $G$ with $\mnorm{\widehat{T}}=n-2f$, 
\begin{align}
    r_G = \max_{\begin{subarray}{c}\widehat{T}\subset G,\\\mnorm{\widehat{T}} = n-2f\end{subarray}} r_{G\widehat{T}}  ~ \leq ~ \epsilon. \label{eqn:r_S_eps}
\end{align}
From~\eqref{eqn:out_S} and~\eqref{eqn:r_S_eps} we obtain that
\begin{align}
    r_S \leq r_G \leq \epsilon. \label{eqn:rsg}
\end{align}

As $\mnorm{G} = n-f$, for every set of agents $T$ with $\mnorm{T} = n-f$, $\mnorm{T \cap G} \geq n-2f$. Therefore, for the set $S$ defined in~\eqref{eqn:out_S}, there exists a subset $\widehat{G}$ of $G$ such that $\widehat{G} \subset S$ and $\mnorm{\widehat{G}} = n-2f$. For such a set $\widehat{G}$, by definition of $r_S$ in~\eqref{eqn:r_S}, we obtain that
\begin{align*}
    r_{S\widehat{G}} \triangleq \dist{x_S}{X_{\widehat{G}}} \leq r_S.
\end{align*}
Substituting from~\eqref{eqn:rsg} above, we obtain that 
\begin{align}
    \dist{x_S}{X_{\widehat{G}}}\leq\epsilon. \label{eqn:dist_bnd_eps}
\end{align}
As $\widehat{G}$ is a subset of $G$, all the agents in $\widehat{G}$ are non-faulty. Therefore, by Assumption~\ref{asp:exists}, $X_{\widehat{G}}$ is a closed set. Recall that $\dist{x_S}{X_{\widehat{G}}} = \inf_{x \in X_{\widehat{G}}} \norm{x_S - x}$. The closedness of $X_{\widehat{G}}$ implies that there exists a point $z \in X_{\widehat{G}}$ such that 
\[\norm{x_S - z} = \inf_{x \in X_{\widehat{G}}} \norm{x_S - x} = \dist{x_S}{X_{\widehat{G}}}.\]
The above, in conjunction with~\eqref{eqn:dist_bnd_eps}, implies that
\begin{align}
    \norm{x_S - z} \leq \epsilon. \label{eqn:pt_1}
\end{align}
Moreover, as $z \in X_{\widehat{G}}$ where $\widehat{G} \subset G$ with $\mnorm{\widehat{G}} = n-2f$ and $\mnorm{G} = n-f$, the $(2f, \, \epsilon)$-redundancy condition stated in Definition~\ref{def:approx_red} implies that $\dist{z}{X_G} \leq \epsilon$.
Similar to an argument made above, under Assumption~\ref{asp:exists}, $X_G$ is a closed set, and therefore, there exists $x^* \in X_G$ such that
\begin{align}
    \norm{z - x^*} = \dist{z}{X_G} \leq \epsilon.  \label{eqn:pt_2}
\end{align}
By triangle inequality,~\eqref{eqn:pt_1} and~\eqref{eqn:pt_2} implies that $\norm{x_S - x^*} \leq \norm{x_S-z} + \norm{z - x^*} \leq 2\epsilon$.
Recall that set $G$ here is an arbitrary set of $n-f$ non-faulty agents. 
\end{proof}




It is worth noting that the algorithm constructed in the proof of Theorem~\ref{thm:suff} only shows sufficiency; it is not a very practical algorithm due to being computationally expensive. \\

In the next part of the paper, i.e., Sections~\ref{sec:grad_des} and~\ref{sec:experiments}, we consider the case when the (non-faulty) agents' cost functions are differentiable. Specifically, we study approximate fault-tolerance in the distributed gradient-descent (DGD) method. 

\section{Distributed Gradient-Descent (DGD) Method}
\label{sec:grad_des}

In this section, we consider a setting wherein the non-faulty agents' cost functions are differentiable. In this particular case, we study the approximate fault-tolerance of the distributed gradient-descent method coupled with a {\em gradient-filter}, described below. We consider the server-based system architecture, shown in Fig.~\ref{fig:sys}, assuming a synchronous system.\\

The DGD method is an iterative algorithm wherein the server maintains an {\em estimate} of a minimum point, and 
updates it iteratively using gradients sent by the agents. Specifically, in each iteration $t \in \{0, \, 1, \ldots \}$, the server starts with an estimate $x^t$ and broadcasts to all the agents. Each non-faulty agent $i$ sends back to the sever the gradient of its cost function at $x^t$, i.e., $\nabla Q_i(x^t)$. However, Byzantine faulty agents may send arbitrary incorrect vectors as their gradients to the server. The initial estimate, named $x^0$, is chosen arbitrarily by the server. \\

A {\em gradient-filter} is a vector function, denoted by $\gf$, that maps the $n$ gradients received by the server from all the $n$ agents to a $d$-dimensional vector, i.e., $\gf: \R^{d \times n} \to \R^d$. For example, an average of all the gradients as in the case of the traditional distributed gradient-descent method is technically a gradient-filter. However, averaging is not quite robust against Byzantine faulty agents~\cite{blanchard2017machine, su2016fault}. The real purpose of a gradient-filter is to mitigate the detrimental impact of incorrect gradients sent by the Byzantine faulty agents. In other words, a gradient-filter {\em robustifies} the traditional gradient-descent method against Byzantine faults. We show that if a gradient-filter satisfies a certain property then it can confer fault-tolerance to the distributed gradient-descent method.\\

We first formally describe below the steps in each iteration of the distributed gradient-descent method implemented on a synchronous server-based system. Note that we constrain the estimates computed by the server to a compact convex set $\W \subset \R^d$. The set $\W$ can be arbitrarily large. For a vector $x \in \R^d$, its projection onto $\W$, denoted by $[x]_\W$, is defined to be                                   
\begin{align}
    [x]_\W = \arg \min_{y \in \W} \norm{x - y}. \label{eqn:proj_W}
\end{align}
As $\W$ is a convex and compact set, $[x]_\W$ is unique for each $x$ (see~\cite{boyd2004convex}).





\subsection{Steps in $t$-th iteration}
\label{sub:steps}

In each iteration $t \in \{0, \,1 , \ldots\}$ the server updates its current estimate $x^t$ to $x^{t+1}$ using Steps {S1} and {S2} described as follows. 

\begin{enumerate}[label=\textbf{S\arabic*:}]
    \item The server requests from each agent the gradient of its local cost function at the current estimate $x^t$. Each non-faulty agent $i$ will then send to the server the gradient $\nabla Q_i(x^t)$, whereas a faulty agent may send an incorrect arbitrary value for the gradient.
    
    The gradient received by the server from agent $i$ is denoted as $g_i^t$. If no gradient is received from some agent $i$, agent $i$ must be faulty (because the system is assumed to be synchronous) -- in this case, the server eliminates the agent $i$ from the system, updates the values of $n$, $f$, and re-assigns the agents indices from $1$ to $n$.
    
    \item \textbf{[Gradient-filtering]} The server applies a gradient-filter $\gf$ to the $n$ received gradients and computes 
    
    $\gf\left(g^t_1, \ldots, \, g^t_n\right) \in \R^d$. Then, the server updates its estimate to
    \begin{equation}
        x^{t+1}=\left[x^t-\eta_t \, \gf\left(g^t_1, \ldots, \, g^t_n\right) \right]_\W \label{eqn:update}
    \end{equation}
    where $\eta_t$ is the step-size of positive value for iteration $t$.
\end{enumerate}

We derive in Theorem~\ref{thm:upper-bound-D} a generic convergence result for the above algorithm. This result is similar to a prior result in \cite[Section 5.2]{bottou1998online}, and will be used to prove the correctness of the algorithm assuming redundancy.

\begin{theorem}
\label{thm:upper-bound-D}
    Consider the update rule~\eqref{eqn:update} in the above iterative algorithm, with diminishing step-sizes $\{\eta_t, \, t = 0, \, 1, \ldots \}$ satisfying $\sum_{t=0}^{\infty}\eta_t=\infty\textrm{ and }\sum_{t=0}^{\infty}\eta_t^2<\infty$.
    Suppose that 
    \[\norm{\gf \left(g^t_1, \ldots, \, g^t_n\right)} < \infty\]
    for all $t$. For some point $x^* \in \W$, if there exists real-valued constants $\D^*\in[0,\max_{x\in\mathcal{W}}\norm{x-x^*})$ and $\xi > 0$ such that for each iteration $t$,
    \begin{equation}
        \phi_t=\iprod{x^t-x^*}{\gf \left(g^t_1, \ldots, \, g^t_n\right)} \geq \xi \textrm{ when }\norm{x^t-x^*} \geq \D^*, \label{eqn:lem_cond}
    \end{equation}
    then $\lim_{t\rightarrow\infty}\norm{x^t - x^*} \leq \D^*$.
\end{theorem}

The values $\D^*$ and $\xi$ in Theorem~\ref{thm:upper-bound-D} may be interdependent. Proof of Theorem~\ref{thm:upper-bound-D} is deferred to Appendix \ref{apdx:lemma-uppr-bnd-D}. \\

Using Theorem~\ref{thm:upper-bound-D} we can obtain conditions under which a gradient-filter guarantees the approximate fault-tolerance property of $(f, \, \epsilon)$-resilience with $\epsilon \geq 0$, of which exact fault-tolerance is a special case. On the other hand, the prior results on the convergence of DGD method with a gradient-filter, e.g., see~\cite{blanchard2017machine, guerraoui2018hidden}, apply  only to exact fault-tolerance.\\

We demonstrate below the utility of Theorem~\ref{thm:upper-bound-D} to obtain the fault-tolerance properties of two commonly used gradient-filters in the literature; namely {\em Comparative Gradient Elimination}~\cite{gupta2020byzantine} and {\em Coordinate-Wise Trimmed Mean}~\cite{su2018finite}.

\subsection{Gradient-Filters and their Fault-Tolerance Properties}
\label{sec:grad_filters}

In this subsection, we present precise approximate fault-tolerance properties of two specific gradient-filters; the Comparative Gradient Elimination (CGE)~\cite{gupta2020byzantine, gupta2019byzantine}, and the Coordinate-Wise Trimmed Mean (CWTM)~\cite{su2018finite, yin2018byzantine}. Note that differentiability of non-faulty agents' cost functions, which is already assumed for the DGD method, implies Assumption~\ref{asp:exists} (see~\cite{boyd2004convex}). We additionally make Assumptions~\ref{assum:lipschitz},~\ref{assum:strongly-convex} and~\ref{assum:compact} about the non-faulty agents' cost functions. Similar assumptions are made in prior work on fault-free distributed optimization~\cite{bertsekas1989parallel, boyd2011distributed, nedic2009distributed}. 


\begin{assumption}[Lipschitz smoothness]
    \label{assum:lipschitz}
    For each non-faulty agent $i$, we assume that the gradient of its cost function $\nabla Q_i(x)$ is Lipschitz continuous, i.e., there exists a finite real value $\mu>0$ such that 
    \begin{equation*}
        \norm{\nabla Q_i(x)-\nabla Q_i(x')}\leq\mu\norm{x-x'}, \quad \forall x,x' \in \W.
    \end{equation*}
\end{assumption}

\begin{assumption}[Strong convexity]
    \label{assum:strongly-convex}
    For a non-empty set of non-faulty agents $\H$, let $Q_{\H}(x)$ denote the average cost function of the agents in $\H$, i.e.,
    \begin{equation*}
        Q_{\H}(x)=\frac{1}{\mnorm{\H}}\sum_{i\in\H}Q_i(x).
    \end{equation*}
    For each such set $\H$ with $\mnorm{\H} = n-f$, we assume that $Q_{\H}(x)$ is strongly convex, i.e., there exists a finite real value $\gamma > 0$ such that
    \begin{equation*}
        \iprod{\nabla Q_{\H}(x)-\nabla Q_{\H}(x')}{x-x'}\geq\gamma\norm{x-x'}^2, \quad \forall x,x' \in \W.
    \end{equation*}
\end{assumption}

Note that, under Assumptions~\ref{assum:lipschitz} and~\ref{assum:strongly-convex}, $\gamma \leq \mu$. This inequality is proved in Appendix~\ref{app:gamma_mu}.
Now, recall that the iterative estimates of the algorithm in Section~\ref{sub:steps} are constrained to a compact convex set $\W \subset \R^d$.

\begin{assumption}[Existence]
    For each set of non-faulty agents $\H$ with $\mnorm{\H} = n-f$, we assume that there exists a point $x_{\H} \in \arg \min_{x \in \R^d} \sum_{i \in \H} Q_i(x)$ such that $x_{\H} \in \W$.
    \label{assum:compact}
\end{assumption}

We describe below the CGE and CWTM gradient-filters. Later, we obtain the fault-tolerance properties of these filters using the result stated in Theorem~\ref{thm:upper-bound-D}, under $(2f, \, \epsilon)$-redundancy.\\

{\bf CGE Gradient-Filter:} To apply the CGE gradient-filter in Step S2, the server sorts the $n$ gradients received from the $n$ agents at the completion of Step S1 as per their Euclidean norms (ties broken arbitrarily):
\begin{equation*}
    \norm{g_{i_1}^t}\leq \ldots \leq \norm{g_{i_{n-f}}^t} \leq \norm{g_{i_{n-f+1}}^t}\leq \ldots \leq \norm{g_{i_n}^t}.
\end{equation*}
That is, the gradient with the smallest norm, $g_{i_1}^t$, is received from agent $i_1$, and the gradient with the largest norm, $g_{i_n}^t$, is received from agent $i_n$. Then, the output of the CGE gradient-filter is the vector sum of the $n-f$ gradients with smallest $n-f$ Euclidean norms. Specifically,
\begin{align}
    \gf\left(g^t_1, \ldots, \, g^t_n \right) = \sum_{j=1}^{n-f}g_{i_j}^t. \label{eqn:cge_gf}
\end{align}
~

{\bf CWTM Gradient-Filter:} To implement this particular gradient-filter in Step S2, the server sorts the $n$ gradients received from the $n$ agents at the completion of Step S1 as per their individual elements. For a vector $v \in \R^d$, we let $v[k]$ denote its $k$-th element. Specifically, for each $k \in \{1, \ldots, \, d\}$, the server sorts the $k$-th elements of the gradients by breaking ties arbitrarily:
\begin{equation*}
    g_{i_1[k]}^t[k] \leq \ldots \leq g_{i_{f + 1}[k]}^t[k] \leq \ldots \leq g_{i_{n-f}[k]}^t[k] \leq  \ldots \leq g_{i_n[k]}^t[k].
\end{equation*}
The gradient with the smallest of the $k$-th element, $g_{i_1[k]}^t$, is received from agent $i_1[k]$, and the gradient with the largest of the $k$-th element, $g_{i_n[k]}^t$, is received from agent $i_n[k]$. For each $k$, the server eliminates the largest $f$ and the smallest $f$ of the $k$-th elements of the gradients received. Then, the output of the CWTM gradient-filter is a vector whose $k$-th element is equal to the average of the remaining $n-2f$ gradients' $k$-th elements. That is, for each $k \in \{1, \ldots, \, d\}$,
\begin{align}
    \gf\left(g^t_1, \ldots, \, g^t_n \right)[k] = \frac{1}{n-2f}\sum_{j=f+1}^{n-f}g_{i_j[k]}^t[k]. \label{eqn:cwtm_gf}
\end{align}
~

We present the precise fault-tolerance properties of the two gradient-filters in Theorems~\ref{thm:conv-guarantee} and~\ref{thm:conv-guarantee_cwtm} below. However, the {\bf reader may skip to Section~\ref{sec:experiments} without loss of continuity}.




Note that, under Assumptions~\ref{assum:strongly-convex} and~\ref{assum:compact}, for each non-empty set of non-faulty agents $\H$ with $\mnorm{\H} = n-f$, the aggregate cost function $\sum_{i \in \H} Q_i(x)$ has a unique minimum point, denoted by $x_{\H}$, in the set $\W$. Specifically,
\begin{align}
    \left\{x_{\H}\right\} = \W \cap \arg \min_{x \in \R^d}\sum_{i \in \H} Q_i(x).  \label{eqn:uniq_min}
\end{align}
We first show below, in Theorem~\ref{thm:conv-guarantee}, that when the fraction of Byzantine faulty agents $f/n$ is bounded then the DGD method with the CGE gradient-filter is $(f, \, \O(\epsilon))$-resilient, under $(2f, \,\epsilon)$-redundancy and the above assumptions. 

\begin{theorem}
    \label{thm:conv-guarantee}
    Suppose that the non-faulty agents' cost functions satisfy the $(2f, \, \epsilon)$-redundancy property, and the Assumptions~\ref{assum:lipschitz},~\ref{assum:strongly-convex} and~\ref{assum:compact} hold true. Consider the algorithm in Section~\ref{sub:steps} with the CGE gradient-filter defined in~\eqref{eqn:cge_gf}. The following holds true:
    
    \begin{enumerate}
        \item $\norm{\gf \left(g^t_1, \ldots, \, g^t_n\right)} < \infty$ for all $t$.
        \item If
            \begin{equation}
                \alpha = 1-\dfrac{f}{n}\left(1+\dfrac{2\mu}{\gamma}\right) > 0 \label{eqn:alpha}
            \end{equation}
            then for each set of $n-f$ non-faulty agents $\H$, for each $\delta > 0$,
            \begin{align*}
                & \phi_t =\iprod{x^t-x_{\H}}{\gf \left(g^t_1, \ldots,  g^t_n\right)} \geq \alpha n \gamma \delta \left( \left( \dfrac{4\mu f}{ \alpha  \gamma} \right) \epsilon + \delta \right) \\
                & \textrm{ when }\norm{x^t-x^*} \geq \left( \dfrac{4\mu f}{ \alpha \,  \gamma} \right)\, \epsilon + \delta. 
            \end{align*}
    \end{enumerate}
    
    
\end{theorem}
~

Proof of Theorem~\ref{thm:conv-guarantee} is deferred to Appendix~\ref{app:proof_1}. 

Let $\H$ denote an arbitrary set of $n-f$ non-faulty agents. If the step-size $\eta_t$ in~\eqref{eqn:update} is diminishing, i.e., $\sum_{t=0}^{\infty}\eta_t=\infty\textrm{ and }\sum_{t=0}^{\infty}\eta_t^2<\infty$, then Theorem~\ref{thm:conv-guarantee}, in conjunction with Theorem~\ref{thm:upper-bound-D} implies that, under the said conditions, 
\begin{equation*}
    \lim_{t\rightarrow\infty}\norm{x^t - x_{\H}} \leq \left( \dfrac{4\mu f}{ \alpha \,  \gamma} \right)\, \epsilon + \delta, \quad \forall \delta > 0.
\end{equation*}
The above implies that $\lim_{t\rightarrow\infty}\norm{x^t - x_{\H}} \leq \left( 4 \mu f/ \alpha \gamma \right) \, \epsilon$. Thus, Theorem~\ref{thm:conv-guarantee} shows that under $(2f, \, \epsilon)$-redundancy, and Assumptions~\ref{assum:lipschitz},~\ref{assum:strongly-convex} and~\ref{assum:compact}, if $\alpha > 0$, or the fraction of Byzantine faulty agents $f/n$ is less than $1/\left(1 + 2(\mu/\gamma) \right)$, then the DGD method with the CGE gradient-filter is asymptotically {\em $(f, \, \D \epsilon)$-resilient} (by Definition~\ref{def:approx_res}) where 
\begin{align}
    \D  = \frac{4 \mu f}{\alpha \gamma} = \frac{4 \mu \, n}{  (n/f) \, \gamma -  ( \gamma  + 2 \mu) }. \label{eqn:D_cge}
\end{align}
A smaller number $f$ of Byzantine faulty agents implies a smaller value of $\D$, and therefore, better fault-tolerance of the algorithm. Moreover, $\D = 0$ when $f = 0$, i.e., the algorithm indeed converges to the actual minimum point of all the agents' aggregate cost function in the fault-free case. Note that under Assumptions~\ref{assum:lipschitz} and~\ref{assum:strongly-convex}, $\gamma \leq \mu$
. So, the fault-tolerance guarantee of the CGE gradient-filter, presented in Theorem~\ref{thm:conv-guarantee}, requires $f/n < 1/3$, or $f < n/3$.




On a side note, we have also obtained an alternative bound for the DGD method with CGE by making better use of the $2f$-redundancy property. The bound is stated in the following theorem, and the proof can be found in Appendix~\ref{appdx:conv-guarantee-2}.
\begin{theorem}
    \label{thm:conv-guarantee-alt}
    Suppose that the non-faulty agents' cost functions satisfy the $(2f, \, \epsilon)$-redundancy property, and the Assumptions~\ref{assum:lipschitz},~\ref{assum:strongly-convex} and~\ref{assum:compact} hold true. Consider the algorithm in Section~\ref{sub:steps} with the CGE gradient-filter defined in~\eqref{eqn:cge_gf}. The following holds true:
    
    \begin{enumerate}
        \item $\norm{\gf \left(g^t_1, \ldots, \, g^t_n\right)} < \infty$ for all $t$.
        \item If $f\leq n/3$ and 
            \begin{equation}
                \alpha = 1-\dfrac{f}{n}\left(1+\dfrac{\mu}{\gamma}\right) > 0 \label{eqn:alpha-alt}
            \end{equation}
            then for each set of $n-f$ non-faulty agents $\H$, for each $\delta > 0$,
            \begin{align*}
                & \phi_t =\iprod{x^t-x_{\H}}{\gf \left(g^t_1, \ldots,  g^t_n\right)} \geq \alpha n \gamma \delta \left( \left( \dfrac{(1+2f)(n-2f)\mu}{ \alpha n \gamma} \right) \epsilon + \delta \right) \\
                & \textrm{ when }\norm{x^t-x^*} \geq \left( \dfrac{(1+2f)(n-2f)\mu}{ \alpha n \gamma} \right)\, \epsilon + \delta. 
            \end{align*}
    \end{enumerate}
\end{theorem}

Next, we show that when the separation between the gradients of the non-faulty agents' cost functions is sufficiently small then the CWTM gradient-filter can guarantee some approximate fault-tolerance under $(2f, \, \epsilon)$-redundancy. To present the fault-tolerance of the CWTM gradient-filter, we make the following additional assumption.

\begin{assumption}
\label{assum:corr_grad}
For two non-faulty agents $i$ and $j$, we assume that there exists $\lambda > 0$ such that for all $x \in \W$,
\[\norm{\nabla Q_i(x) - \nabla Q_j(x)} \leq \lambda \max \left\{ \norm{\nabla Q_i(x)},  \norm{\nabla Q_j(x)}\right\}.\]
\end{assumption}
Due to the triangle triangle inequality, Assumption~\ref{assum:corr_grad} trivially holds true when $\lambda = 2$. However, we can presently guarantee fault-tolerance of CWTM gradient-filter when $\lambda < \gamma/(\mu \sqrt{d})$ where $\mu$ and $\gamma$ are the Lipschitz smoothness and strong convexity coefficients, respectively defined in Assumption~\ref{assum:lipschitz} and~\ref{assum:strongly-convex}. Recall the definition of point $x_{\H} \in \R^d$ from~\eqref{eqn:uniq_min} where $\H$ denotes an arbitrary set of $n-f$ non-faulty agents.

\begin{theorem}
\label{thm:conv-guarantee_cwtm}
    Suppose that the non-faulty agents' cost functions satisfy the $(2f, \, \epsilon)$-redundancy property, and the Assumptions~\ref{assum:lipschitz},~\ref{assum:strongly-convex},~\ref{assum:compact} and~\ref{assum:corr_grad} hold true. Consider the algorithm in Section~\ref{sub:steps} with the CWTM gradient-filter defined in~\eqref{eqn:cwtm_gf}. The following holds true:
    \begin{enumerate}
        \item $\norm{\gf \left(g^t_1, \ldots, \, g^t_n\right)} < \infty$ for all $t$.
        \item If $\lambda < \gamma/(\mu \sqrt{d})$ then for each set of $n-f$ non-faulty agents $\H$, for each $\delta > 0$,
    \begin{align*}
        \phi_t & =\iprod{x^t-x_{\H}}{\gf \left(g^t_1, \ldots, g^t_n\right)} \geq \left( 2 \sqrt{d} n \mu \lambda \epsilon +  \left(\gamma -  \sqrt{d} \lambda \mu \right)  \delta \right) \delta 
    \end{align*}
    when \(
        \norm{x^t - x_{\H}} \geq \cfrac{2 \sqrt{d} n \mu \lambda}{(\gamma - \sqrt{d} \mu \lambda)} \, \epsilon  + \delta.
    \)
    \end{enumerate}
    
\end{theorem}
~

Proof of Theorem~\ref{thm:conv-guarantee_cwtm} is deferred to Appendices~\ref{app:proof_2}. 

By similar arguments as in the case of CGE, Theorem~\ref{thm:conv-guarantee_cwtm}, in conjunction with Theorem~\ref{thm:upper-bound-D}, implies that the DGD method with CWTM gradient-filter and diminishing step-sizes is asymptotically {\em $(f, ~ \D' \, \epsilon )$-resilient} where
\[\D' = \frac{2 \sqrt{d} n \mu \lambda}{(\gamma - \sqrt{d} \mu \lambda)} = \left( \frac{2 n }{(\gamma/ \mu \lambda \sqrt{d}) - 1} \right),\]
under the conditions stated in Theorem~\ref{thm:conv-guarantee_cwtm}.
The smaller the value of $\lambda$ is, i.e., the closer non-faulty gradients to each other are, the smaller is the value of $\D'$, and therefore, better is the approximate fault-tolerance guarantee of the CWTM gradient-filter. Unlike the CGE gradient-filter, resilience of CWTM presented in Theorem~\ref{thm:conv-guarantee_cwtm} is independent of $f$, as long as $\lambda < \gamma/(\mu \sqrt{d})$. However, the condition on $\lambda$ to guarantee the resilience of CWTM gradient-filter depends upon the dimension $d$ of the optimization problem. Larger dimension result in a tighter bound on $\lambda$.

\section{Numerical Experiments}
\label{sec:experiments}

We present simulation results for an empirical evaluation of the CGE and CWTM gradient-filters applied to the problem of {\em distributed linear regression}~\cite{amemiya1985advanced}. More details about the experiments can be found in Appendix~\ref{apdx:experiments}. \\

We consider the synchronous server-based system in Figure~\ref{fig:sys}. We assume that $n = 6$ and $f = 1$. Each agent $i$ knows a row vector $A_i$ of dimension $d = 2$. Each agent $i$ makes a real-valued (scalar) observation denoted by $B_i$ such that $B_i=A_ix^*+N_i$, where $x^* = (1,1)^T$ for all $i$, and $N_i$ is a randomly chosen noise. The value of $A_i$, $B_i$ and $N_i$ are omitted here for brevity. 
To solve the linear regression problem distributedly, each agent $i$'s cost function is defined as $Q_i(x)=\left(B_i-A_ix\right)^2$. For a non-empty set of agents $S$, we denote by $A_S$ a matrix of dimension $\mnorm{S} \times 2$ obtained by stacking rows $\{A_i, \, i \in S\}$.
Similarly, we obtain column vector $B_S$ by stacking the values $\{B_i, \, i \in S\}$. Thus for every such non-empty set $S$, $Q_S(x) \triangleq \sum_{i\in S}\left(B_i-A_ix\right)^2=\norm{B_S-A_Sx}^2$. 
The rows $A_1, \ldots, \, A_n$ are chosen specifically to ensure that the system has $2f$-redundancy if $N_i = 0$ for all $i$. 
That is, each matrix $A_S$ with $|S|\geq n-2f=4$ is column full-rank or $\rank{A_S}=d=2$. Consequentially, the cost function $Q_S(x)$ has a unique minimum point for each set $S$ with $|S|\geq4$. \\



We simulate the distributed gradient-descent algorithm described in Section~\ref{sec:grad_des} by assuming agent 1 to be Byzantine faulty, i.e., the set of honest agents is $\mathcal{H}=\{2,3,4,5,6\}$. The minimum point of $\sum_{i\in\H}Q_i(x)$, denoted by $x_\H$, can be obtained by solving $B_\H=A_\H x$. Specifically, $x_\H=(1.0780, 0.9825)^T$. The goal of fault-tolerant distributed linear regression is to estimate $x_\H$. 
In our simulations, it can be verified that the agents' cost functions satisfy the {\em $(2f, \, \epsilon)$-redundancy} property, stated in Definition~\ref{def:approx_red}, with $\epsilon = 0.0890$. It can also be verified that the non-faulty agents' cost functions satisfy Assumptions~\ref{assum:lipschitz} and~\ref{assum:strongly-convex} with $\mu=2$ and $\gamma=0.712$, respectively. We simulate the following fault behaviors for the Byzantine agents:
\begin{itemize}
    \item \textit{gradient-reverse}: the faulty agent \textit{reverses} its true gradient. Suppose the correct gradient of a faulty agent $i$ at step $t$ is $s_i^t$, the agent $i$ will send the incorrect gradient $g_i^t=-s_i^t$ to the server.
    \item \textit{random}: the faulty agent sends a randomly chosen vector in $\mathbb{R}^d$. In our experiments, the faulty agent in each step chooses i.i.d. Gaussian random vector with mean 0 and an isotropic covariance matrix of standard deviation $200$.
\end{itemize}

In the simulations, we apply a diminishing step size $\eta_t$, and a convex compact $\W$ as described in previous sections. For comparison purpose, all experiments have the same initial estimate $x^0=(-0.0085,-0.5643)^T$. In every execution, the estimates practically converge after 400 iterations. We document the output of the algorithm to be $x_{\textrm{out}}=x^{500}$. The outputs for the two gradient-filters, CGE and CWTM, under different faulty behaviors, are shown in Table~\ref{tab:results}. Note that $\dist{x_\H}{x_{\textrm{out}}}=\norm{x_\H-x_{\textrm{out}}}$. In all the executions, the distance $\norm{x_\H - x_{\textrm{out}}} < \epsilon$. \\

\comment{Email on Aug 11: On Figure 2: why doesn't "plain GD" curve show any variations at all. \\
\\Plotted again.}

For the said executions, we plot in Figure~\ref{fig:fault-comparison} the values of the aggregate cost function $\sum_{i \in \H}Q_i(x^t)$ (referred as {\em loss}) and the approximation error $\norm{x^t - x_{\H}}$ (referred as {\em distance}) for iteration $t$ ranging from $0$ to $1500$. We also show the plots of the fault-free DGD method where the faulty agent is omitted, 
and the DGD method without any gradient-filter when agent $1$ is Byzantine faulty. The details for iteration $t$ ranging from 0 to 80 are also highlighted in Figure~\ref{fig:fault-comparison-detail}. \\

We also conducted experiments for distributed learning with support vector machine with faulty agents in the distributed learning system (see Section~\ref{sub:app}). We observed that the DGD method with the said gradient-filters reaches comparable performance to the fault-free case, and that DGD cannot reach convergence if it uses plain averaging to aggregate the gradients, including the faulty ones. We also observed that the accuracy of the learning process depends upon the correlation between the data points of non-faulty agents. Details of those results can be found in Appendix~\ref{appdx:exp-learning}.

\begin{table}[tb!]
    \centering
    \caption{{\it For the distributed linear regression problem, our algorithm's outputs with gradient-filters CGE and CWTM, and the approximation errors, corresponding to executions when the faulty agent $1$ exhibits two different types of Byzantine faults; \emph{gradient-reverse} and \emph{random}. 
    }}
    \small
    \begin{tabular}{c|cc|cc}
         & \multicolumn{2}{c}{{\bf gradient-reverse}} & \multicolumn{2}{c}{{\bf random}} \\
         & $x_{\mathrm{out}}$ & $\dist{x_\H}{x_{\mathrm{out}}}$ & $x_{\mathrm{out}}$ & $\dist{x_\H}{x_{\mathrm{out}}}$ \\
        \hline
        {\bf CGE} & $\begin{pmatrix}1.0541\\0.9826\end{pmatrix}$ & 0.0239& $\begin{pmatrix}1.0779\\0.9826\end{pmatrix}$ & $4.72\times10^{-5}$ \\
        {\bf CWTM} & $\begin{pmatrix}1.0645\\0.9924\end{pmatrix}$ & 0.0167 & $\begin{pmatrix}1.0775\\0.9840\end{pmatrix}$ & $1.51\times10^{-3}$ \\
    \end{tabular}
\label{tab:results}
\end{table}

\begin{figure}[t]
    \centering
    \includegraphics[width=.65\linewidth]{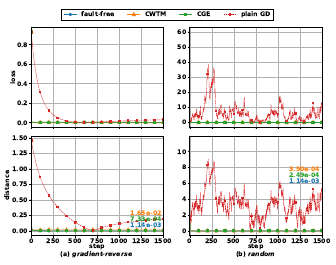}
    \caption{    
    \small{\it The losses, i.e., $\sum_{i \in \H}Q_i(x^t)$, and distances, i.e., $\norm{x^t - x_{\H}}$, versus the number of iterations in the algorithm. The final approximation errors, i.e., $\norm{x^{5000} - x_{\H}}$, are annotated in the same colors as their corresponding plots. For the executions shown, agent $1$ is assumed to be Byzantine faulty. The different columns show the results when the faulty agent exhibits the different types of faults: \emph{(a)}~gradient-reverse, and \emph{(b)}~random. Apart from the plots with CGE (in {\em green}) and CWTM (in {\em yellow}) gradient-filters, we also plot the fault-free DGD method where the faulty agent is omitted 
    (in {\em blue}), and the DGD method without any gradient-filters when agent $1$ is Byzantine faulty (in {\em red}), both using averaging for aggregation.
    }}
    \label{fig:fault-comparison}
\end{figure}
\begin{figure}[t]
    \centering
    \includegraphics[width=.65\linewidth]{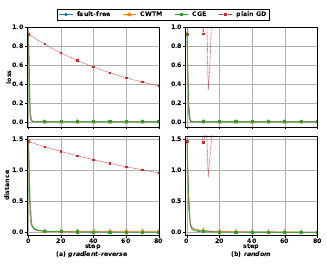}
    \caption{\small{\it The losses, i.e., $\sum_{i \in \H}Q_i(x^t)$, and distances, i.e., $\norm{x^t - x_{\H}}$, versus the number of iterations in the algorithm, magnified for the initial 80 iterations in the training process. The interpretation of the plots is same as that in Figure~\ref{fig:fault-comparison}.
    }}
    \label{fig:fault-comparison-detail}
\end{figure}

\section{Summary}
We have considered the problem of {\em approximate} Byzantine fault-tolerance -- a generalization of the {\em exact} fault-tolerance problem studied in prior work~\cite{gupta2020fault}. Unlike the exact fault-tolerance, the goal in approximate fault-tolerance is to design a distributed optimization algorithm that approximates a minimum point of the aggregate cost function of (at least $n-f$) non-faulty agents, in the presence of up to $f$ (out of $n$) Byzantine faulty agents. We have defined approximate fault-tolerance formally as $(f, \, \epsilon)$-resilience where $\epsilon \in \R_{\geq 0}$ represents the approximation error. In the first part of the paper, i.e, Section~\ref{sec:approx-fault-tolerance}, we have obtained necessary and sufficient conditions for achieving $(f, \, \epsilon)$-resilience. 
In the second part of the paper, i.e., Sections~\ref{sec:grad_des} and~\ref{sec:experiments}, we have considered the case when agents' cost functions are differentiable. In this particular case, we have obtained a generic approximate fault-tolerance property of the distributed gradient-descent method equipped with Byzantine robust gradient aggregation or {\em gradient-filter}, and have demonstrated the utility of this property by considering two specific well-known gradient-filters; comparative gradient elimination and coordinate-wise trimmed mean. In Section~\ref{sec:experiments}, we have demonstrated the applicability of our results through experiments.

\section*{Acknowlegements}
Research reported in this paper was supported in part by the Army Research Laboratory under Cooperative Agreement W911NF- 17-2-0196, and by the National Science Foundation award 1842198. The views and conclusions contained in this document are those of the authors and should not be interpreted as representing the official policies, either expressed or implied, of the Army Research Laboratory, National Science Foundation or the U.S. Government. Research reported in this paper is also supported in part by a Fritz Fellowship from Georgetown University.


\bibliographystyle{acm}
\bibliography{ref}





\appendix
\section{Appendix: Proof of Lemma~\ref{lem:n_2f}}
\label{app:n_2f}

\begin{lemma*}[Restated]
If $f \geq n/2$ then there cannot exist a deterministic $(f, \, \epsilon)$-resilient algorithm for any $\epsilon \geq 0$.
\end{lemma*}

\begin{proof}
We prove the lemma by contradiction. We consider a case when $n = 2$, $d = 1$, i.e., $Q_i: \R \to \R$ for all $i$, and all the cost functions have unique minimum points. Suppose that $f = 1$, and that there exists a deterministic $(f, \, \epsilon)$-resilient algorithm $\Pi$ for some $\epsilon \geq 0$. Without loss of generality, we suppose that agent $2$ is Byzantine faulty. We denote $x_1 = \arg \min_{x \in \R}Q_1(x)$. \\

The Byzantine agent $2$ can choose to behave as a non-faulty agent with cost function $\widetilde{Q}_2(x) = Q_1(x - x_1 - 2\epsilon - \delta)$ where $\delta$ is some positive real value. Now, note that the minimum point of $\widetilde{Q}_2(x)$, which we denoted by $x_2$, is unique and equal to $x_1 + 2\epsilon + \delta$. Therefore, $\mnorm{x_1 - x_2} = 2 \epsilon + \delta > 2 \epsilon$. As the identity of Byzantine agent is a priori unknown, the server cannot distinguish between scenarios; (i) agent $1$ is non-faulty, and (ii) agent $2$ is non-faulty. Now, being deterministic algorithm, $\Pi$ should produce the same output in both the scenarios. In scenario (i), as $\Pi$ is assumed $(f, \, \epsilon)$-resilient, its output must lie in the interval $[x_1 - \epsilon, \, x_1 + \epsilon]$. Similarly, in scenario (ii), the output of $\Pi$ must lie in the interval $[x_2 - \epsilon, \, x_2 + \epsilon]$. However, as $\mnorm{x_1 - x_2} > 2\epsilon$, the two intervals  $[x_1 - \epsilon, \, x_1 + \epsilon]$ and $[x_2 - \epsilon, \, x_2 + \epsilon]$ do not overlap. Therefore, $\Pi$ cannot be $(f, \, \epsilon)$-resilient in both the scenarios simultaneously, which is a contradiction to the assumption that $\Pi$ is $(f, \, \epsilon)$-resilient. \\

The above argument extends easily for the case when $n > 2$, and $f > n/2$.
\end{proof}

\section{Appendix: The Special Case of $(f, \, 0)$-Resilience}
\label{app:exact_approx}
We show that {\em $(f, \, 0)$-resilience}, stated in Definition~\ref{def:approx_res}, and {\em exact fault-tolerance}, defined in Section~\ref{sub:background}, are equivalent in the deterministic framework. Specifically, we show that a deterministic $(f, \, 0)$-resilient algorithm achieves exact fault-tolerance, and a deterministic exact fault-tolerant algorithm is $(f, \, 0)$-resilient. 
We consider the server-based system architecture. Notation $\arg \min_{x \in \R^d}$ is simply written as $\arg \min$, unless otherwise stated.\\

First, we show that $(f, \, 0)$-resilience implies exact fault-tolerance. Suppose that there exists a deterministic $(f, \, 0)$-resilient algorithm $\Pi$. Consider an arbitrary execution $E_{\H}$ of $\Pi$ wherein $\H \subseteq \{1, \ldots, \, n\}$ denotes the set of all the non-faulty agents, and let $\widehat{x}$ denote the output. Recall that, as there are at most $f$ faulty agents, $\mnorm{\H} \geq n-f$. To prove that $\Pi$ has exact fault-tolerance, it suffices to show that, in execution $E_{\H}$, $\widehat{x}$ is a minimum point of the aggregate cost function of all non-faulty agents $\sum_{i \in \H} Q_i(x)$.\\

By Definition~\ref{def:approx_res} of $(f, \, 0)$-resilience, for every set $S \subseteq \H$ with $\mnorm{S} = n-f$,
\begin{align*}
    \widehat{x} \in \arg \min \sum_{i \in S} Q_i(x). 
\end{align*}
Therefore, for every set $S$ with $S \subseteq \H$ and $\mnorm{S} = n-f$,
\begin{align}
    \sum_{i \in S} Q_i(\widehat{x}) \leq  \sum_{i \in S} Q_i(x), \quad \forall x \in \R^d. \label{eqn:e_a_1}
\end{align}
Now, note that there are ${\mnorm{\H}} \choose {n-f}$ subsets in $\H$ of size $n-f$, and each agent $i \in \H$ is contained in $\mnorm{\H} - 1 \choose n-f-1$ of those subsets. Therefore,
\begin{align}
    \sum_{\begin{subarray}{c}S \subseteq \H,\\ \mnorm{S} = n-f \end{subarray}} \sum_{i \in S} Q_i(x) = {\mnorm{\H} - 1 \choose n-f-1} \, \sum_{i \in \H} Q_i(x). \label{eqn:sum_choose_sum}
\end{align}
Substituting from~\eqref{eqn:e_a_1} in~\eqref{eqn:sum_choose_sum} we obtain that 
\begin{align*}
    \sum_{i \in \H} Q_i(\widehat{x}) \leq \sum_{i \in \H} Q_i(x), \quad \forall x \in \R^d
\end{align*}
The above implies that
\begin{align*}
    \widehat{x} \in \arg \min \sum_{i \in \H} Q_i(x).
\end{align*}
The above proves that $\Pi$ has exact fault-tolerance in execution $E_{\H}$.\\

Now, we show that exact fault-tolerance implies $(f, \, 0)$-resilience. Suppose that $\Pi$ is a deterministic algorithm with exact fault-tolerance. Similar to above, consider an arbitrary execution $E_{\H}$ of $\Pi$ wherein set $\H$ comprises all the non-faulty agents, and $\widehat{x}$ is its output. Therefore, 
\begin{align*}
    \widehat{x} \in \arg \min \sum_{i \in \H} Q_i(x). 
\end{align*}
To prove that $\Pi$ is $(f, \, 0)$-resilient, it suffices to show that in execution $E_{\H}$ for every set $S \subseteq \H$ with $\mnorm{S} = n-f$, $\widehat{x}$ is a minimum point of the aggregate cost function $\sum_{i \in S} Q_i(x)$. This is trivially true when $\mnorm{\H} = n-f$. We assume below that $\mnorm{\H} > n-f$.\\

Consider an arbitrary subset $S$ of $\H$ with $\mnorm{S} = n-f$. Consider an execution $E_{S}$ wherein $S$ is the set of all non-faulty agents, with the remaining agents in $\{1, \ldots, \, n\} \setminus S$ being Byzantine faulty. Suppose that the inputs from all the agents to the server in $E_S$ are identical to their inputs in $E_{\H}$. Therefore, as $\Pi$ is a deterministic algorithm, its output in execution $E_{S}$ is same as that in execution $E_{\H}$, i.e., $\widehat{x}$. Moreover, as $\Pi$ is assumed to have exact fault-tolerance, 
\[\widehat{x} \in \arg \min \sum_{i \in S} Q_i(x).\]
As $S$ is an arbitrary subset $S$ of $\H$ with $\mnorm{S} = n-f$, the above proves that $\Pi$ is $(f, \, 0)$-resilient in execution $E_{\H}$. 

\section{Appendix: Proof of $\gamma \leq \mu$}
\label{app:gamma_mu}

We show below that if Assumptions~\ref{assum:lipschitz} and~\ref{assum:strongly-convex} hold true simultaneously then $\gamma \leq \mu$.\\

Consider an arbitrary set of $n-f$ non-faulty agents $\H$, and two arbitrary non-identical points $x, \, y \in \R^d$, i.e., $x \neq y$. If Assumption~\ref{assum:lipschitz} holds true then
\begin{align*}
    \norm{\nabla Q_i(x) - \nabla Q_i(y)} \leq \mu \norm{x - y}, \quad \forall i \in \H. 
\end{align*}
Therefore, owing to the Cauchy-Schwartz inequality, for all $i \in \H$,
\begin{align}
    \iprod{x - y}{\nabla Q_i(x) - \nabla Q_i(y)} \leq \norm{x - y} \, \norm{\nabla Q_i(x) - \nabla Q_i(y)} \leq  \mu \norm{x - y}^2. \label{eqn:gm_1}
\end{align}
From~\eqref{eqn:gm_1} we obtain that
\begin{align}
    \sum_{i \in \H} \iprod{x - y}{\nabla Q_i(x) - \nabla Q_i(y)} \leq  \mu \mnorm{\H} \, \norm{x - y}^2. \label{eqn:gm_3}
\end{align}
If Assumption~\ref{assum:strongly-convex} holds true then
\begin{align}
    \sum_{i \in \H} \iprod{x - y}{\nabla Q_i(x) - \nabla Q_i(y)} \geq  \gamma \mnorm{\H} \, \norm{x - y}^2. \label{eqn:gm_2}
\end{align}
As $x, \,y$ are arbitrary non-identical points,~\eqref{eqn:gm_3} and~\eqref{eqn:gm_2} together imply that $\gamma \leq \mu$.

\section{Appendix: Proof of Theorem~\ref{thm:upper-bound-D}}
\label{apdx:lemma-uppr-bnd-D}

The proof of Theorem~\ref{thm:upper-bound-D} relies on the following sufficient criterion for the convergence of non-negative sequences.\\

\noindent \fbox{\begin{minipage}{\linewidth}
\begin{lemma}[Bottou, 1998 \cite{bottou1998online}]
    \label{lem:seq_conv}
    Consider a sequence of real values $\{u_t, \, t = 0, \, 1, \ldots \}$. If $u_t \geq 0, \, \forall t$ then 
    \begin{align}
        \sum_{t = 0}^\infty (u_{t+1} - u_t)_{+} < \infty \implies \left\{\begin{array}{c} u_t \underset{t \to \infty}{\longrightarrow} u_\infty < \infty \\ \\ \sum_{t = 0}^\infty (u_{t+1} - u_t)_{-} > -\infty \end{array}\right.
    \end{align}
    where the operators $(\cdot)_{+}$ and $(\cdot)_{-}$ are defined as follows for a real scalar $x$,
    \begin{align*}
        (x)_{+} = \left\{\begin{array}{ccc} x &, & x > 0\\ 0 &, & \text{otherwise} \end{array}\right. \text{, and } (x)_{-} = \left\{\begin{array}{ccc} 0 &, & x > 0\\ x &, & \text{otherwise} \end{array}\right.
    \end{align*}
\end{lemma}
\end{minipage}}
~\\

Recall from the statement of Theorem~\ref{thm:upper-bound-D} that $x^* \in \W$ where $\W$ is a compact convex set.
We define, for all $t \in \{0, \, 1, \ldots \}$,
\begin{align}
    e_t = \norm{x^t-x^*}. \label{eqn:def_et}
\end{align}
Next, we define a univariate real-valued function $\psi: \R \to \R$:
\begin{equation}
    \psi(y)=\left\{\begin{array}{cc}
        0, & y < (\D^*)^2, \\
        \left(y-(\D^*)^2\right)^2,  & y \geq (\D^*)^2.
    \end{array}\right.
    \label{eqn:def_psi}
\end{equation}
Let $\psi'(y)$ denote the derivative of $\psi$ at $y \in \R$. Specifically, 
\begin{equation}
    \psi'(y)=\max\left\{0, 2\left(y-(\D^*)^2\right)\right\}.
    \label{eqn:psi_x}
\end{equation}
We show below that $\psi'(\cdot)$ is a Lipschitz continuous function with Lipschitz coefficient of $2$. From~\eqref{eqn:psi_x}, we obtain that
\begin{equation}
    \mnorm{\psi'(y) - \psi'(z)} = \left\{\begin{array}{ccc} 2 \mnorm{y - z} & , &  \text{ both } y, \, z \geq (\D^*)^2\\ 2 \mnorm{y - (D^*)^2} & , & y \geq (\D^*)^2, \, z < (\D^*)^2 \\0 & , & \text{ both } y, \, z < (\D^*)^2 \end{array}\right. \label{eqn:lip_psi_1}
\end{equation}
Note from~\eqref{eqn:lip_psi_1} that for the case when $y \geq (\D^*)^2, \, z < (\D^*)^2$,  
\[\mnorm{\psi'(y) - \psi'(z)} = 2 \mnorm{y - (D^*)^2} < 2 \mnorm{y - z}.\]
Similarly, due to symmetry, when $y < (\D^*)^2, \, z \geq (\D^*)^2$ then $\mnorm{\psi'(y) - \psi'(z)} = 2 2 \mnorm{z - (D^*)^2} < 2 \mnorm{y - z}.$ 
Therefore, from~\eqref{eqn:psi_x} we obtain that
\begin{align}
    \mnorm{\psi'(y) - \psi'(z)} \leq 2 \mnorm{y - z}, \quad \forall y, \, z \in \R. \label{eqn:psi_lipschitz}
\end{align}
The Lipschitz continuity of $\psi'(\cdot)$, shown in~\eqref{eqn:psi_lipschitz}, implies that~\cite[Section 4.1]{bottou2018optimization}
\begin{equation}
    \psi(z)-\psi(y)\leq(z-y)\psi'(y)+(z-y)^2, \quad \forall y, \, z \in \R.
    \label{eqn:psi_bnd}
\end{equation}\\


Now, for each $t \in \{0, \, 1, \ldots \}$, we define 
\begin{equation}
    h_t=\psi(e^2_t).
    \label{eqn:ht_def_noise}
\end{equation}
From~\eqref{eqn:psi_bnd} and~\eqref{eqn:ht_def_noise}, for all $t$, we obtain that
\begin{equation*}
    h_{t+1} - h_t = \psi \left( e_{t+1}^2\right) - \psi \left( e_t^2\right) \leq \left(e_{t+1}^2 - e_t^2\right) \cdot \psi'\left( e_t^2\right) + \left( e_{t+1}^2 - e_t^2 \right)^2.
\end{equation*}
From now on we use $\psi_t'$ as a shorthand for $\psi'(e_t^2)$. From above we have
\begin{equation}
    h_{t+1} - h_t \leq\left(e_{t+1}^2 - e_t^2\right)\psi'_t + \left( e_{t+1}^2 - e_t^2 \right)^2.
    \label{eqn:ht_1_noise}
\end{equation}
~

Now, recall from~\eqref{eqn:update} that for all $t \in \{0, \, 1, \ldots\}$,
\begin{equation}
    x^{t+1}=\left[x^t-\eta_t\, \gf \left(g^t_1, \ldots, \, g^t_n\right)\right]_{\W}
    \label{eqn:appendix-update}
\end{equation}
By the non-expansion property of Euclidean projection onto a closed convex set, 
\begin{equation*}
    \norm{x^{t+1}-x^*}\leq\norm{x^t-x^*-\eta_t \, \gf \left(g^t_1, \ldots, \, g^t_n\right) }.
\end{equation*}
Recall from~\eqref{eqn:def_et} that $e_t$ denotes $\norm{x^t - x^*}$ for all $t$. Upon squaring the both sides in the above inequality, we obtain that
\begin{align}
    e_{t+1}^2\leq e_t^2&-2\eta_t\iprod{x_t-x^*}{\gf \left(g^t_1, \ldots, \, g^t_n\right)} \nonumber \\
        &+ \eta_t^2 \norm{\gf \left(g^t_1, \ldots, \, g^t_n\right)}^2. \label{eqn:above_phi_t}
\end{align}
Recall, from~\eqref{eqn:lem_cond} in the statement of Theorem~\ref{thm:upper-bound-D}, that 
\[\phi_t = \iprod{x_t-x^*}{\gf \left(g^t_1, \ldots, \, g^t_n\right)}, \quad \forall t.\] 
Substituting from the above in~\eqref{eqn:above_phi_t}, we obtain that
\begin{equation}
    e_{t+1}^2\leq e_t^2-2\eta_t\phi_t+\eta_t^2\norm{\gf \left(g^t_1, \ldots, \, g^t_n\right)}^2.
    \label{eqn:proj-bound}
\end{equation}
As $\psi'_t\geq0$, $\forall t$, substituting from \eqref{eqn:proj-bound} in \eqref{eqn:ht_1_noise} we get
\begin{equation}
    h_{t+1}-h_t\leq\left(-2\eta_t\phi_t+\eta_t^2\norm{\gf \left(g^t_1, \ldots, \, g^t_n\right)}^2\right)\psi_t'+\left(e_{t+1}^2-e_t^2\right)^2.
    \label{eqn:ht_2_noise}
\end{equation}
~

Note that for an arbitrary $t$,
\begin{equation}
    \mnorm{e_{t+1}^2-e_t^2}=(e_{t+1}+e_t)\mnorm{e_{t+1}-e_t}.
    \label{eqn:a2-b2}
\end{equation}
As $\mathcal{W}$ is assumed compact, there exists $\Gamma=\max_{x\in\mathcal{W}}\norm{x-x^*} < \infty$. We assume $\Gamma>0$, otherwise $\W = \{x^*\}$ and the theorem is trivial. Recall from the update rule~\eqref{eqn:update}, which is re-stated above in~\eqref{eqn:appendix-update}, that $x^t\in\mathcal{W}$ for all $t$, and that $x^* \in \W$. Therefore,
\begin{equation}
    e_t = \norm{x^t - x^*} \leq \max_{x\in\mathcal{W}}\norm{x-x^*} = \Gamma, \quad \forall t. \label{eqn:et_gamma}
\end{equation}
From~\eqref{eqn:et_gamma}, for all $t$, we obtain that
\begin{equation*}
    e_{t+1}+e_t\leq2\Gamma.
\end{equation*}
Substituting from above in \eqref{eqn:a2-b2} implies that
\begin{equation}
    \mnorm{e_{t+1}^2-e_t^2}\leq2\Gamma\mnorm{e_{t+1}-e_t}, \quad \forall t.
    \label{eqn:ab}
\end{equation}
From triangle inequality, we get
\begin{equation*}
    \mnorm{e_{t+1}-e_t}=\mnorm{\norm{x^{t+1}-x^*}-\norm{x^t-x^*}}\leq\norm{x^{t+1}-x^t}.
\end{equation*}
Substituting from above in~\eqref{eqn:ab} we obtain that
\begin{equation}
    \mnorm{e_{t+1}^2-e_t^2}\leq2\Gamma\norm{x^{t+1}-x^t}.
    \label{eqn:ab-2}
\end{equation}
Due to the non-expansion property of Euclidean projection onto a closed convex set, from~\eqref{eqn:appendix-update} we obtain that
\begin{align*}
    \norm{x^{t+1}-x^t}&=\norm{\left[x^t-\eta_t \, \gf \left(g^t_1, \ldots, \, g^t_n\right)\right]_\mathcal{W}-x^t} \leq \eta_t\norm{\gf \left(g^t_1, \ldots, \, g^t_n\right)}.
\end{align*}
Substituting from above in~\eqref{eqn:ab-2} we obtain that
\begin{equation*}
    \mnorm{e_{t+1}^2-e_t^2}\leq2\eta_t\Gamma\norm{\gf \left(g^t_1, \ldots, \, g^t_n\right)}.
\end{equation*}
Thus, 
\begin{equation}
    \left(e_{t+1}^2-e_t^2\right)^2\leq4\eta_t^2 \, \Gamma^2 \, \norm{\gf \left(g^t_1, \ldots, \, g^t_n\right)}^2.
    \label{eqn:ab-f}
\end{equation}
Substituting from~\eqref{eqn:ab-f} in~\eqref{eqn:ht_2_noise} we obtain that, for all $t$,
\begin{align}
     h_{t+1}-h_t&\leq\left(-2\eta_t\phi_t+\eta_t^2\norm{\gf \left(g^t_1, \ldots, \, g^t_n\right)}^2\right)\psi_t' + 4\eta_t^2 \Gamma^2 \norm{\gf \left(g^t_1, \ldots, \, g^t_n\right)}^2 \nonumber \\
     &=-2\eta_t\phi_t\psi_t'+\eta_t^2\left(\psi_t'+4\Gamma^2\right)\norm{\gf \left(g^t_1, \ldots, \, g^t_n\right)}^2.
     \label{eqn:ht-ht-noise}
\end{align}
~

Recall from~\eqref{eqn:et_gamma} that $e_t \leq \Gamma$. Also, by assumption, $\D^* < \max_{x \in \W} \norm{x - x^*} = \Gamma$. Recall that $\psi'_t$ is short for $\psi'(e^2_t)$. Therefore, from~\eqref{eqn:psi_x} we obtain that 
\begin{equation}
    0\leq\psi_t'\leq2\left(\Gamma^2-\left(\D^*\right)^2\right)\leq2\Gamma^2, \quad \forall t.
    \label{eqn:d-psi-bound}
\end{equation}
As the statement of Theorem~\ref{thm:upper-bound-D} assumes that $\norm{\gf \left(g^t_1, \ldots, \, g^t_n\right)} < \infty$ for all $t$, there exists a real value $\M < \infty$ such that
\begin{equation}
    \norm{\gf \left(g^t_1, \ldots, \, g^t_n\right)} \leq \M, \quad \forall t.
    \label{eqn:m-bound}
\end{equation}
Substituting from~\eqref{eqn:d-psi-bound} and~\eqref{eqn:m-bound} in \eqref{eqn:ht-ht-noise} we obtain that
\begin{equation}
    h_{t+1}-h_t\leq-2\eta_t\phi_t\psi_t'+6\eta_t^2\Gamma^2\M^2.
    \label{eqn:ht_M}
\end{equation}
We now use Lemma~\ref{lem:seq_conv} to prove that $h_\infty=0$ as follows. \\

For an iteration $t$, we consider below two possible cases: (i) $e_t < \D^*$, and (ii) $e_t = \D^* + \delta$ for some $\delta \geq 0$.
\begin{enumerate}[label=Case \roman*), nosep]
    \item In this particular case, $\psi_t'=0$. Therefore, due to Cauchy-Schwartz inequality,
    \begin{align*}
        \mnorm{\phi_t}&=\mnorm{\iprod{x^t-x^*}{\gf \left(g^t_1, \ldots, \, g^t_n\right)}} \leq e_t \, \norm{\gf \left(g^t_1, \ldots, \, g^t_n\right)}.
    \end{align*}
    Substituting from~\eqref{eqn:m-bound} above, we obtain that $\mnorm{\phi_t}\leq\Gamma\M<\infty$. Therefore, 
    \begin{equation}
        \phi_t\psi_t'=0.
        \label{eqn:case-1}
    \end{equation}
    \item In this particular case, from~\eqref{eqn:psi_x}, we obtain that
    \[\psi_t' =  2 \left( (\D^* + \delta)^2 - (\D^*)^2\right) = 2 \delta (2 \D^* + \delta).\] 
    Now, by assumption, $\phi_t \geq \xi$ when $e_t \geq \D^*$ where $\xi > 0$. Therefore, 
    \begin{equation}
        \phi_t\psi_t' \geq 2 \xi \delta (2 \D^* + \delta) .
        \label{eqn:case-2}
    \end{equation}
\end{enumerate}
From~\eqref{eqn:case-1} and~\eqref{eqn:case-2} above, we obtain that
\begin{equation}
    \phi_t\psi_t'\geq0, \quad \forall t. \label{eqn:positive_phi_psi}
\end{equation}
Substituting the above in~\eqref{eqn:ht_M} implies that
\begin{equation*}
    h_{t+1}-h_t \leq 6\eta_t^2\Gamma^2\M^2, \quad \forall t.
\end{equation*}
Recall that notation $(\cdot)_+$ from Lemma \ref{lem:seq_conv}. The above inequality implies that
\begin{equation*}
    \left(h_{t+1}-h_t\right)_+\leq6\eta_t^2\Gamma^2\M^2.
\end{equation*}
As $\sum_{t=0}^\infty\eta_t^2<\infty, ~ \forall t$, and constants $\Gamma, \, \M<\infty$, the above implies that
\begin{equation*}
    \sum_{t=0}^\infty\left(h_{t+1}-h_t\right)_+\leq6\Gamma^2\M^2 \sum_{t=0}^\infty\eta_t^2<\infty.
\end{equation*}
As $h_t\geq0$ for all $t$, the above in conjunction with Lemma~\ref{lem:seq_conv} implies that
\begin{align}
\begin{split}
    h_t\xrightarrow[]{t\rightarrow\infty}h_\infty<\infty, ~\textrm{and} \\
    \sum_{t=0}^\infty\left(h_{t+1} - h_t\right)_->-\infty. \label{eqn:conv-noise}
\end{split}
\end{align}
Note that $h_\infty-h_0=\sum_{t=0}^\infty\left(h_{t+1}-h_t\right)$. Thus, from \eqref{eqn:ht_M} we obtain that
\begin{equation}
    h_\infty-h_0\leq-2\sum_{t=0}^\infty\eta_t\phi_t\psi_t'+6\Gamma^2\M^2 \, \sum_{t=0}^\infty\eta_t^2. \label{eqn:h_t_h_0_bnd}
\end{equation}
Therefore, form above we obtain that
\begin{equation}
    2\sum_{t=0}^\infty\eta_t\phi_t\psi_t' \leq h_0- h_{\infty} + 6\Gamma^2\M^2 \, \sum_{t=0}^\infty\eta_t^2 . \label{eqn:h_t_h_0_bnd_2}
\end{equation}
By assumption, $\sum_{t=0}^\infty\eta_t^2 < \infty$. From~\eqref{eqn:conv-noise}, $0\leq h_{\infty} < \infty$. Substituting from~\eqref{eqn:et_gamma} that $e_t < \infty, \, \forall t$ in Definition of $h_t$~\eqref{eqn:ht_def_noise}, we obtain that $h_0 = \psi (e^2_0) < \infty$. Therefore,~\eqref{eqn:h_t_h_0_bnd_2} implies that
\begin{equation*}
    2\sum_{t=0}^\infty\eta_t\phi_t\psi_t' < h_0 + 6\Gamma^2\M^2 \, \sum_{t=0}^\infty\eta_t^2 <\infty .
\end{equation*}
Or simply,
\begin{equation}
    \sum_{t=0}^\infty\eta_t\phi_t\psi_t' < \infty. \label{eqn:contra_0}
\end{equation}

Finally, we reason below by contradiction that $h_\infty=0$. 
Note that for any $\zeta>0$, there exists a unique positive value $\beta$ such that $\zeta = 2\beta\left(2\D^*+\sqrt{\beta}\right)^2$.
Suppose that $h_\infty = 2\beta(2\D^*+\sqrt{\beta})^2$ for some positive value $\beta$. As the sequence $\left\{h_t\right\}_{t=0}^\infty$ converges to $h_\infty$ (see~\eqref{eqn:conv-noise}), there exists some finite $\tau \in \mathbb{Z}_{\geq0}$ such that for all $t \geq \tau$,
\begin{equation*}
    \mnorm{h_t-h_\infty}\leq\beta\left(2\D^*+\sqrt{\beta}\right)^2 \implies h_t \geq h_\infty - \beta\left(2\D^*+\sqrt{\beta}\right)^2.
\end{equation*}
As $h_\infty = 2\beta(2\D^*+\sqrt{\beta})^2$, the above implies that
\begin{equation}
    h_t\geq\beta\left(2\D^*+\sqrt{\beta}\right)^2, \quad ~\forall t\geq\tau.
    \label{eqn:contra}
\end{equation}
Therefore (cf. \eqref{eqn:def_psi} and \eqref{eqn:ht_def_noise}), for all $t\geq\tau$,
\begin{gather*}
    \left(e_t^2-\left(\D^*\right)^2\right)^2\geq\beta\left(2\D^*+\sqrt{\beta}\right)^2,~\textrm{or} \\
    \mnorm{e_t^2-\left(\D^*\right)^2}\geq\sqrt{\beta}\left(2\D^*+\sqrt{\beta} \right).
\end{gather*}
Thus, for each $t \geq \tau$, either
\begin{equation}
    e_t^2\geq\left(\D^*\right)^2+\sqrt{\beta}\left(2\D^*+\sqrt{\beta}\right)=\left(\D^*+\sqrt{\beta}\right)^2,
    \label{eqn:et-opt1}
\end{equation}
or
\begin{equation}
    e_t^2\leq\left(\D^*\right)^2-\sqrt{\beta}\left(2\D^*+\sqrt{\beta}\right) < (\D^*)^2.
    \label{eqn:et-opt2}
\end{equation}
If the latter, i.e. \eqref{eqn:et-opt2}, holds true for some $t' \geq \tau$ then
\begin{equation*}
    h_{t'}=\psi\left(e^2_{t'}\right)=0,
\end{equation*}
which contradicts \eqref{eqn:contra}. Therefore,~\eqref{eqn:contra} implies~\eqref{eqn:et-opt1} for all $t \geq \tau$.\\

From above we obtain that if $h_\infty=2\beta(2\D^*+\sqrt{\beta})^2$ then there exists $\tau<\infty$ such that for all $t\geq\tau$,
\begin{equation*}
    e_t\geq\D^*+\sqrt{\beta}.
\end{equation*}
Thus, from~\eqref{eqn:case-2} we obtain that 
\begin{equation*}
    \phi_t\psi'_t \geq 2 \xi \sqrt{\beta} (2 \D^* + \sqrt{\beta}), \quad \forall t \geq \tau.
\end{equation*}
Therefore, 
\[\sum_{t=\tau}^\infty\eta_t\phi_t\psi_t' \geq 2 \xi \sqrt{\beta} (2 \D^* + \sqrt{\beta}) \sum_{t=\tau}^\infty\eta_t = \infty.\]
This is a contradiction to~\eqref{eqn:contra_0}. Therefore, $h_\infty=0$, and by definition of $h_t$ in~\eqref{eqn:ht_def_noise}, 
\begin{equation*}
    h_{\infty} = \lim_{t\rightarrow\infty} \psi(e^2_t) = 0.
\end{equation*}
Hence, by definition of $\psi(\cdot)$ in~\eqref{eqn:def_psi},
\begin{align*}
    \lim_{t\rightarrow\infty}\norm{x^t-x^*} \leq \D^*.
\end{align*}

\section{Appendix: Lemma~\ref{lemma:upbd-vector}} 
\label{app:lemma}

\begin{lemma}
    \label{lemma:upbd-vector}
    For $p$ vectors $a_1,...,a_p$ and a positive integer $q\leq p/2$, if there exists some real number $r\geq0$, such that for any set $S\subset[p]$ with $\mnorm{S}=q$ we have
    \begin{align}
        \norm{\sum_{i\in S}a_i}\leq r,
        \label{eqn:lemma-condition}
    \end{align}
    then for all $i\in[p]$ we also have
    \begin{align}
        \norm{a_i}\leq 2r.
    \end{align}
\end{lemma}
Let us first introduce some notations. For two vectors $a,b$, denote by $[a]_b$ the vector projection of $a$ onto $b$. Also, denote by $\boldsymbol{e}_b$ the unit vector in the direction of $b$; in other words, $b = \norm{b}\cdot\boldsymbol{e}_b$.
\begin{proof}
    The result is trivially true for $q=1$, as the size of set $S$ is 1. In the remainder of this proof, we assume that $q>1$, and consequently, $p>3$. The proof is by contradiction. 
    
    Suppose for some $k$, $\norm{a_k}>2r$. Define $s$ such that $a_k = s\cdot\boldsymbol{e}_{a_k}$. We have $s=\mnorm{s}>2r$. By \eqref{eqn:lemma-condition}, for any set $S'\subset[p]\backslash\{k\}$ with $\mnorm{S'}=q-1$, we have
    \begin{align}
        \norm{\sum_{i\in S'}a_i + a_{k}}\leq r.
    \end{align}
    Due to the non-expansion property of vector projection, 
    \begin{align}
        \norm{\left[\sum_{i\in S'}a_i + a_{k}\right]_{a_{k}}} 
        \leq \norm{\sum_{i\in S'}a_i+a_{k}}\leq r.
    \end{align}
    Let $t$ be a real value such that $\left[\sum_{i\in S'}a_i\right]_{a_{k}} = t\cdot \boldsymbol{e}_{a_k}$. It follows that 
    \begin{align}
        \norm{\left[\sum_{i\in S'}a_i + a_{k}\right]_{a_{k}}} = \norm{\left[\sum_{i\in S'}a_i\right]_{a_{k}}+a_{k}} = \norm{(t + s)\boldsymbol{e}_{a_k}} = \mnorm{t+s} \leq r.
    \end{align}
    Therefore, $t+s\leq r$. So $t\leq r-s<r-2r=-r$.

    On the other hand, for any $j\in[p]$, $j\notin S'$ and $j\neq k$, by \eqref{eqn:lemma-condition}, we also have
    \begin{align}
        \norm{\sum_{i\in S'}a_i + a_j}\leq r.
    \end{align}
    Similarly, due to the non-expansion property of vector projection,
    \begin{align}
        \norm{\left[\sum_{i\in S'}a_i + a_j\right]_{a_{k}}}\leq \norm{\sum_{i\in S'}a_i+a_j}\leq r.
    \end{align}
    Let $t'$ be a real value such that ${\left[a_j\right]_{a_{k}}} = t'\cdot \boldsymbol{e}_{a_k}$. It follows that 
    \begin{align}
        \norm{\left[\sum_{i\in S'}a_i + a_j\right]_{a_{k}}} = \norm{(t+t')\boldsymbol{e}_{a_k}} = \mnorm{t+t'}\leq r.
        \label{eqn:lemma-5-1}
    \end{align}
    Since $t<-r$, we have $t'>0$.

    Since $q\leq p/2$, there are at least $q$ different element that are in $[p]$ but not in $S'\cup\{k\}$. Let $q-1$ of such elements form a set $S''$. For each $i\in S''$, let $t'_i$ be a real value such that ${\left[a_i\right]_{a_{k}}} = t'_i\cdot \boldsymbol{e}_{a_k}$. By \eqref{eqn:lemma-5-1}, $t'_i>0$ for all $i\in S''$. Also, recall that $s=\norm{a_k}>2r\geq0$. Therefore, we have
    \begin{align}
        \norm{\sum_{i\in S''}a_i+a_{k}}&\geq \norm{\left[\sum_{i\in S''}a_i+a_{k}\right]_{a_{k}}} = \norm{\left[\sum_{i\in S''}a_i\right]_{a_{k}}+a_{k}} =
        \norm{\left(\sum_{i\in S''}t'_i+s\right)\cdot\boldsymbol{e}_{a_{k}}}
        \nonumber \\
        & =\mnorm{\sum_{i\in S''}t'_i + s} = \sum_{i\in S''}t'_i + s > \sum_{i\in S''}t'_i + 2r > 2r.
    \end{align}
    However, since $\mnorm{S''\cup\{k\}}=q$, by \eqref{eqn:lemma-condition}, 
    \begin{align}
        \norm{\sum_{i\in S''}a_i+a_{k}}\leq r.
    \end{align}
    This is a contradiction. Hence the proof.
\end{proof}

\section{Appendix: Lemma~\ref{lemma:bounding-gradients-at-xH} on the upper bound of the norm of gradients with $(2f,\epsilon)$-redundancy}

We assume $f > 0$ to ignore the trivial case of $f = 0$. We also assume $f\leq n/3$. 

Consider an arbitrary set $\H$ of non-faulty agents with $\mnorm{\H} = n-f$. Recall that under Assumptions~\ref{assum:strongly-convex} and~\ref{assum:compact}, the aggregate cost function $\sum_{i \in \H} Q_i(x)$ has a unique minimum point in set $\W$, which we denote by $x_{\H}$. Specifically, 
\begin{align}
    \left\{x_{\H} \right\} = \W \cap \arg \min_{x \in \R^d} \sum_{i \in \H} Q_i(x). \label{eqn:x_H_W}
\end{align}

\begin{lemma}
    \label{lemma:bounding-gradients-at-xH}
    Let $\H$ be a set of non-faulty agents with $\mnorm{\H}=n-f$. Suppose a group of $n$ cost functions $Q_i(x)$'s satisfies $(2f,\epsilon)$-redundancy, and Assumptions~\ref{assum:lipschitz},~\ref{assum:strongly-convex}, and~\ref{assum:compact} hold. Also, suppose $f\leq n/3$.
    Then, in any execution, for any subset $T\subset\H$ of $f$ agents, we have
    \begin{align}
        \norm{\sum_{j\in T}\nabla Q_j(x_{\H})} \leq (n-2f)\mu\epsilon.
        \label{eqn:lipschitz-distance-S1-complement}
    \end{align}
    Also, for any $j\in\H$ we have
    \begin{align}
        \norm{\nabla Q_j(x_\H)} \leq 2(n-2f)\mu\epsilon.
    \label{eqn:honest-norm-bound}
    \end{align}
\end{lemma}

\begin{proof}
    Consider any subset $S_1 \subset \H$ with $\mnorm{S_1} = n-2f$. From triangle inequality, $\forall x \in \R^d$,
\begin{align*}
    \norm{\sum_{j\in S_1}\nabla Q_j(x)-\sum_{j\in S_1}\nabla Q_j(x_{\H})} \leq \sum_{j \in S_1}\norm{\nabla Q_j(x) - \nabla Q_j(x_{\H})}.
\end{align*}
Under Assumption~\ref{assum:lipschitz}, i.e., Lipschitz continuity of non-faulty gradients, for each non-faulty agent $j$, $\norm{\nabla Q_j(x) - \nabla Q_j(x_{\H})} \leq \mu \norm{x - x_{\H}}$. Substituting this above implies that 
\begin{equation}
    \norm{\sum_{j\in S_1}\nabla Q_j(x)-\sum_{j\in S_1}\nabla Q_j(x_{\H})} \leq \mnorm{S_1} \mu \, \norm{x-x_{\H}}.
    \label{eqn:lipschitz-distance}
\end{equation}
As $\mnorm{S_1} = n-2f$, the $(2f,\epsilon)$-redundancy property defined in Definition~\ref{def:approx_red} and the fact \eqref{eqn:x_H_W} that $x_{\H}$ is the unique minimum of $\sum_{i\in\H}Q_i(x)$ in $\W$ imply that for all $x_1\in\arg\min_x\sum_{j\in S_1}Q_j(x)$,
\[\norm{x_1-x_{\H}} \leq \epsilon.\]
Substituting from above in~\eqref{eqn:lipschitz-distance} implies that, for all $x_1\in\arg\min_x\sum_{j\in S_1}Q_j(x)$,
\begin{equation}
    \norm{\sum_{j\in S_1}\nabla Q_j(x_1)-\sum_{j\in S_1}\nabla Q_j(x_{\H})} \leq \mnorm{S_1} \mu \, \norm{x_1-x_{\H}} \leq \mnorm{S_1} \mu \epsilon.
    \label{eqn:lipschitz-distance-2}
\end{equation}
For all $x_1\in\arg\min_x\sum_{j\in S_1}Q_j(x)$, $\sum_{j\in S_1}\nabla Q_j(x_1) = 0$. Thus,~\eqref{eqn:lipschitz-distance-2} implies that
\begin{equation}
    \norm{\sum_{j\in S_1}\nabla Q_j(x_{\H})} \leq \mnorm{S_1} \mu \epsilon.
    \label{eqn:lipschitz-distance-S1}
\end{equation}
On the other hand, for $x_{\H}$ we have $\sum_{j\in\H}\nabla Q_j(x_{\H})=0$. Thus, with $\mnorm{S_1}=n-2f$, \eqref{eqn:lipschitz-distance-S1} implies that 
\begin{align}
    \norm{\sum_{j\in\H\backslash S_1}\nabla Q_j(x_{\H})} = \norm{\sum_{j\in\H}\nabla Q_j(x_{\H}) - \sum_{j\in S_1}\nabla Q_j(x_{\H})} = \norm{\sum_{j\in S_1}\nabla Q_j(x_{\H})} \leq (n-2f)\mu\epsilon.
\end{align}
Note that $\mnorm{\H\backslash S_1}=f$. Since the choices of $\H$ and $S_1$ are arbitrary, the above implies that for any subset of agents $T\subset\H$ with $\mnorm{T}=f$, we have
\begin{align*}
    \norm{\sum_{j\in T}\nabla Q_j(x_{\H})} \leq (n-2f)\mu\epsilon.
\end{align*}
Since we assumed that $f\leq n/3$, $f\leq (n-f)/2$. By Lemma~\ref{lemma:upbd-vector}, letting $p=n-f$ and $q=f$, from above we have for all $j\in\H$,
\begin{align*}
    \norm{\nabla Q_j(x_{\H})} \leq 2(n-2f)\mu\epsilon.
\end{align*}
\end{proof}

\section{Appendix: Proof of Theorem~\ref{thm:conv-guarantee}}
\label{app:proof_1}

Throughout, we assume $f > 0$ to ignore the trivial case of $f = 0$.\\

Consider an arbitrary set $\H$ of non-faulty agents with $\mnorm{\H} = n-f$. Recall that under Assumptions~\ref{assum:strongly-convex} and~\ref{assum:compact}, the aggregate cost function $\sum_{i \in \H} Q_i(x)$ has a unique minimum point in set $\W$, which we denote by $x_{\H}$. Specifically, 
\begin{align}
    \left\{x_{\H} \right\} = \W \cap \arg \min_{x \in \R^d} \sum_{i \in \H} Q_i(x). \label{eqn:x_H_W}
\end{align}
Recall from~\eqref{eqn:cge_gf} that for CGE gradient-filter, in update rule~\eqref{eqn:update},
\begin{align}
    \gf\left(g^t_1, \ldots, \, g^t_n \right) = \sum_{j=1}^{n-f}g_{i_j}^t, \quad \forall t. \label{eqn:cge_grad_filter_proof}
\end{align}
~

First, we show that $\norm{\sum_{j=1}^{n-f}g_{i_j}^t} < \infty, ~ \forall t$. 
For all $x$ and $i\in\H$, by Assumption \ref{assum:lipschitz},
\begin{equation*}
    \norm{\nabla Q_i(x)-\nabla Q_i(x_{\H})}\leq\mu\norm{x-x_{\H}}.
\end{equation*}
By triangle inequality,
\begin{equation*}
    \norm{\nabla Q_i(x)}\leq\norm{\nabla Q_i(x_{\H})}+\mu\norm{x-x_{\H}}.
\end{equation*}
Substituting from~\eqref{eqn:honest-norm-bound} in Lemma~\ref{lemma:bounding-gradients-at-xH} above, we obtain that
\begin{equation}
    \norm{\nabla Q_i(x)}\leq2(n-2f)\mu\epsilon+\mu\norm{x-x_{\H}}\leq2n\mu\epsilon+\mu\norm{x-x_{\H}}.
    \label{eqn:honest-bound-everywhere}
\end{equation}
We use the above inequality~\eqref{eqn:honest-bound-everywhere} to show below that $\norm{\sum_{j=1}^{n-f}g_{i_j}^t}$ is bounded for all $t$. Recall that for each iteration $t$, 
\begin{equation*}
    \norm{g_{i_1}^t}\leq...\leq\norm{g_{i_{n-f}}^t}\leq\norm{g_{i_{n-f+1}}^t}\leq...\leq\norm{g_{i_n}^t}.
\end{equation*}
As there are at most $f$ Byzantine agents, for each $t$ there exists $\sigma_t\in\H$ such that
\begin{equation}
    \norm{g_{i_{n-f}}^t}\leq\norm{g_{i_{\sigma_t}}^t}.
    \label{eqn:honest-bound}
\end{equation}
As $g_j^t=\nabla Q_j(x^t)$ for all $j\in\H$, from~\eqref{eqn:honest-bound} we obtain that
\begin{equation*}
    \norm{g_{i_j}^t}\leq\norm{\nabla Q_{\sigma_t}(x^t)}, \quad \forall j \in \{1, \ldots, n-f\}, ~ t.
\end{equation*}
Substituting from~\eqref{eqn:honest-bound-everywhere} above we obtain that for every $j \in \{1, \ldots, n-f\}$,
\begin{equation*}
    \norm{g_{i_j}^t}\leq\norm{g_{i_{n-f}}^t}\leq2n\mu\epsilon+\mu\norm{x^t-x_{\H}}.
\end{equation*}
Therefore, from triangle inequality,
\begin{equation}
    \norm{\sum_{j=1}^{n-f}g_{i_j}^t}\leq\sum_{j=1}^{n-f}\norm{g_{i_j}^t}\leq(n-f)\left(2n\mu\epsilon+\mu\norm{x^t-x_{\H}}\right). \label{eqn:filtered-upperbound}
\end{equation}
Recall from~\eqref{eqn:x_H_W} that $x_{\H} \in \W$. Let $\Gamma = \max_{x \in \W} \norm{x - x_{\H}}$. As $\W$ is a compact set, $\Gamma < \infty$. Recall from the update rule~\eqref{eqn:update} that $x^t \in \W$ for all $t$. Thus, $\norm{x^t - x_{\H}} \leq \max_{x \in \W} \norm{x - x_{\H}} = \Gamma < \infty$. Substituting this in~\eqref{eqn:filtered-upperbound} implies that
\begin{equation}
    \norm{\sum_{j=1}^{n-f}g_{i_j}^t} \leq (n-f) \left( 2n \mu \epsilon + \mu \Gamma\right) < \infty. 
\end{equation}
Recall that in this particular case, $\sum_{j=1}^{n-f}g_{i_j}^t = \gf\left(g^t_1, \ldots, \, g^t_n \right)$ (see~\eqref{eqn:cge_grad_filter_proof}). Therefore, from above we obtain that
\begin{align}
    \norm{\gf\left(g^t_1, \ldots, \, g^t_n \right)} < \infty, \quad \forall t. \label{eqn:cge_bnd_grd}
\end{align}
~


Next, we show that for an arbitrary $\delta>0$, 
\begin{align}
    \phi_t\triangleq\iprod{x^t-x_{\H}}{\sum_{j=1}^{n-f}g_{i_j}^t} \geq \alpha n\gamma\delta\left(\cfrac{4\mu f}{\alpha\gamma}\epsilon+\delta\right)>0 ~ \text{ when } ~ \norm{x^t-x_{\H}} \geq \cfrac{4\mu f}{\alpha\gamma}\epsilon + \delta, \nonumber
\end{align}
where $\alpha$ is defined by \eqref{eqn:alpha}.


Consider an arbitrary iteration $t$. Note that, as $\mnorm{\H} = n-f$, there are at least $n-2f$ agents that are common to both sets $\H$ and $\{i_1,...,i_{n-f}\}$. We let $\H^t = \{i_1,...,i_{n-f}\} \cap \H$. The remaining set of agents $\B^t = \{i_1,...,i_{n-f}\} \setminus \H^t$ comprises of only faulty agents. Note that $\mnorm{\H^t} \geq n-2f $ and $\mnorm{\B^t} \leq f$. Therefore,
\begin{equation}
    \phi_t=\iprod{x^t-x_{\H}}{\sum_{j\in\H^t}g_j^t}+\iprod{x^t-x_{\H}}{\sum_{k\in\B^t}g_k^t}.
    \label{eqn:phi-t-two-parts}
\end{equation}
Consider the first term in the right-hand side of~\eqref{eqn:phi-t-two-parts}. Note that
\begin{align*}
    \iprod{x^t-x_{\H}}{\sum_{j\in\H^t}g_j^t}&=\iprod{x^t-x_{\H}}{\sum_{j\in\H^t}g_j^t+\sum_{j\in\H\backslash\H^t}g_j^t-\sum_{j\in\H\backslash\H^t}g_j^t} \nonumber \\
    &=\iprod{x^t-x_{\H}}{\sum_{j\in\H}g_j^t}-\iprod{x^t-x_{\H}}{\sum_{j\in\H\backslash\H^t}g_j^t}.
\end{align*}
Recall that $g_j^t=\nabla Q_j(x^t)$, $\forall j\in\H$. Therefore,
\begin{align}
    \iprod{x^t-x_{\H}}{\sum_{j\in\H^t}g_j^t}=&\iprod{x^t-x_{\H}}{\sum_{j\in\H}\nabla Q_j(x^t)} -\iprod{x^t-x_{\H}}{\sum_{j\in\H\backslash\H^t}\nabla Q_j(x^t)}. \label{eqn:first_phi_1}
\end{align}
~

\noindent Due to the strong convexity assumption (i.e., Assumption~\ref{assum:strongly-convex}), for all $x, \, y \in \W$,
\begin{equation}
    \iprod{x- y}{\nabla \sum_{j\in\H}Q_j(x)-\nabla\sum_{j\in\H} Q_j(y)} \geq \mnorm{\H}\, \gamma\norm{x-y}^2.
    \label{eqn:assum-strong-convex-restate}
\end{equation}
By the projection on $\W$ in update rule \eqref{eqn:update}, $x^t\in\W$ for all $t$. By \eqref{eqn:x_H_W}, $x_{\H}$ is also in $\W$. As $x_{\H}$ is minimum point of $\sum_{j \in \H}Q_j(x)$, $\nabla \sum_{j \in \H} Q_j(x_{\H}) = 0$. Thus, by \eqref{eqn:assum-strong-convex-restate} it follows that
\begin{align}
    \iprod{x^t-x_{\H}}{\sum_{j\in\H}\nabla Q_j(x^t)} & = \iprod{x^t-x_{\H}}{\nabla\sum_{j\in\H}Q_j(x^t)-\nabla\sum_{j\in\H}Q_j(x_{\H})} \geq \mnorm{\H} \, \gamma\norm{x^t-x_{\H}}^2.
    \label{eqn:inner-prod-h}
\end{align}
Now, due to the Cauchy-Schwartz inequality, 
\begin{align}
    \iprod{x^t-x_{\H}}{\sum_{j\in\H\backslash\H^t}\nabla Q_j(x^t)}&=\sum_{j\in\H\backslash\H^t}\iprod{x^t-x_{\H}}{\nabla Q_j(x^t)} \leq\sum_{j\in\H\backslash\H^t}\norm{x^t-x_{\H}}\, \norm{\nabla Q_j(x^t)}.
    \label{eqn:inner-prod-h-ht}
\end{align}
Substituting from~\eqref{eqn:inner-prod-h} and~\eqref{eqn:inner-prod-h-ht} in~\eqref{eqn:first_phi_1} we obtain that
\begin{equation}
    \iprod{x^t-x_{\H}}{\sum_{j\in\H^t}g_j^t} \geq \gamma\mnorm{\H}\, \norm{x^t-x_{\H}}^2-\sum_{j\in\H\backslash\H^t}\norm{x^t-x_{\H}}\, \norm{\nabla Q_j(x^t)}.
    \label{eqn:phi-t-first-part}
\end{equation}
~

Next, we consider the second term in the right-hand side of~\eqref{eqn:phi-t-two-parts}. From the Cauchy-Schwartz inequality, 
\begin{equation*}
    \iprod{x^t-x_{\H}}{g_k^t}\geq-\norm{x^t-x_{\H}}\, \norm{g_k^t}.
\end{equation*}
Substituting from~\eqref{eqn:phi-t-first-part} and above in~\eqref{eqn:phi-t-two-parts} we obtain that
\begin{equation}
    \phi_t\geq\gamma\mnorm{\H}\, \norm{x^t-x_{\H}}^2-\sum_{j\in\H\backslash\H^t}\norm{x^t-x_{\H}}\, \norm{\nabla Q_j(x^t)}-\sum_{k\in\B^t}\norm{x^t-x_{\H}}\, \norm{g_k^t}.
    \label{eqn:phi-t-2}
\end{equation}

Recall that, due to the sorting of the gradients, for an arbitrary $k \in \B^t$ and an arbitrary $j \in \H\backslash\H^t$,
\begin{equation}
    \norm{g_k^t}\leq\norm{g_j^t}=\norm{\nabla Q_j(x^t)}. \label{eqn:k_B_j_H}
\end{equation}
Recall that $\B^t = \{i_1, \ldots, \, i_{n-f}\} \setminus \H^t$. Thus, $\mnorm{\B^t} = n-f-\mnorm{\H^t}$. Also, as $\mnorm{\H} = n-f$, $\mnorm{\H\backslash\H^t} = n- f - \mnorm{\H^t}$. That is, $\mnorm{\B^t} = \mnorm{\H\backslash\H^t}$. Therefore,~\eqref{eqn:k_B_j_H} implies that
\begin{align*}
    \sum_{k \in \B^t} \norm{g_k^t} \leq \sum_{j \in \H\backslash\H^t} \norm{\nabla Q_j(x^t)}.
\end{align*}
Substituting from above in~\eqref{eqn:phi-t-2}, we obtain that
\begin{align*}
    \phi_t&\geq\gamma\mnorm{\H}\, \norm{x^t-x_{\H}}^2-2\sum_{j\in\H\backslash\H^t}\norm{x^t-x_{\H}}\, \norm{\nabla Q_j(x^t)}.
\end{align*}
Substituting from \eqref{eqn:honest-bound-everywhere}, i.e., $\norm{\nabla Q_i(x)}\leq2n\mu\epsilon+\mu\norm{x-x_{\H}}$, above we obtain that
\begin{align*}
    \phi_t&\geq\gamma\mnorm{\H}\, \norm{x^t-x_{\H}}^2 - 2\mnorm{\H\backslash\H^t}\, \norm{x^t-x_{\H}}\, (2n\mu\epsilon+\mu\norm{x^t-x_{\H}}) \nonumber\\
    & \geq \left(\gamma\mnorm{\H}-2\mu\mnorm{\H\backslash\H^t}\right)\norm{x^t-x_{\H}}^2 - 4n\mu\epsilon \mnorm{\H\backslash\H^t}\, \norm{x^t-x_{\H}}.
\end{align*}
As $\mnorm{\H} = n-f$ and $\mnorm{\H\backslash\H^t}\leq f$, the above implies that
\begin{align}
\begin{split}
\label{eqn:phi-t-pre-final}
    \phi_t & \geq\left(\gamma(n-f)-2\mu f\right)\norm{x^t-x_{\H}}^2-4n\mu\epsilon f\, \norm{x^t-x_{\H}} \\
        & = \left(\gamma(n-f)-2\mu f\right)\norm{x^t-x_{\H}}\, \left(\norm{x^t-x_{\H}}-\dfrac{4n\mu\epsilon f}{\gamma(n-f)-2\mu f} \right) \\
        & = n \gamma \, \left( 1 - \frac{f}{n} \left( 1 + \frac{2\mu}{\gamma}\right)\right) \norm{x^t-x_{\H}} \, \left(\norm{x^t-x_{\H}} - \frac{4 \mu f \, \epsilon}{\gamma\left( 1 - \frac{f}{n}\left(1 + \frac{2 \mu}{\gamma} \right)\right)} \right).
\end{split}
\end{align}
Recall from~\eqref{eqn:alpha} that
$$\alpha = 1-\dfrac{f}{n}\left(1+\dfrac{2\mu}{\gamma}\right). $$
Substituting from above in~\eqref{eqn:phi-t-pre-final} we obtain that
\begin{align}
    \phi_t \geq \alpha \, n \gamma \,\norm{x^t-x_{\H}}\, \left(\norm{x^t-x_{\H}}- \frac{4\mu f}{\alpha\gamma} \epsilon \right).
    \label{eqn:phi-t-final}
\end{align}
As it is assumed that $\alpha > 0$,~\eqref{eqn:phi-t-final} implies that for an arbitrary $\delta>0$,
\begin{align*}
    \phi_t \geq \alpha n \gamma \delta\left(\cfrac{4\mu f}{\alpha\gamma}\epsilon+\delta\right) > 0 ~ \text{ when } \norm{x^t-x_{\H}} \geq \cfrac{4\mu f}{\alpha\gamma}\epsilon+\delta.
\end{align*}
Hence, the proof.

\section{Appendix: Proof of Theorem~\ref{thm:conv-guarantee-alt}}
\label{appdx:conv-guarantee-2}

The result in Theorem~\ref{thm:conv-guarantee-alt} makes better use of the $2f$-redundancy property than the previous proof in Appendix~\ref{app:proof_1}. 
    The proof of Theorem~\ref{thm:conv-guarantee-alt} has a different second half from the proof in the Appendix~\ref{app:proof_1} after \eqref{eqn:cge_bnd_grd}. The remainder of the proof is replaced by the following analysis:\\

\hrule
~\\

Next, we show that for an arbitrary $\delta>0$, 
\begin{align}
    \phi_t\triangleq\iprod{x^t-x_{\H}}{\sum_{j=1}^{n-f}g_{i_j}^t} \geq \alpha n\gamma\delta\left(\dfrac{(1+2f)(n-2f)\mu}{ \alpha n \gamma}\epsilon+\delta\right)>0 ~ \text{ when } ~ \norm{x^t-x_{\H}} \geq \dfrac{(1+2f)(n-2f)\mu}{ \alpha n \gamma}\epsilon + \delta, \nonumber
\end{align}
where $\alpha$ is defined by \eqref{eqn:alpha-alt}.

Consider an arbitrary iteration $t$. Note that, as $\mnorm{\H} = n-f$, there are at least $n-2f$ agents that are common to both sets $\H$ and $\{i_1,...,i_{n-f}\}$. We let $\H^t$ be an arbitrary set such that $\H^t \subset \{i_1,...,i_{n-f}\} \cap \H$ with $\mnorm{\H^t}=n-2f$. The remaining set of agents $\B^t = \{i_1,...,i_{n-f}\} \setminus \H^t$ may contain up to $f$ Byzantine agents. Note that $\mnorm{\H^t} = n-2f $ and $\mnorm{\B^t} =f$. We have
\begin{equation}
    \phi_t=\iprod{x^t-x_{\H}}{\sum_{j\in\H^t}g_j^t}+\iprod{x^t-x_{\H}}{\sum_{k\in\B^t}g_k^t}.
    \label{eqn:phi-t-two-parts-alt}
\end{equation}
Consider the first term on the right-hand side of~\eqref{eqn:phi-t-two-parts-alt}. Note that
\begin{align*}
    \iprod{x^t-x_{\H}}{\sum_{j\in\H^t}g_j^t}&=\iprod{x^t-x_{\H}}{\sum_{j\in\H^t}g_j^t+\sum_{j\in\H\backslash\H^t}g_j^t-\sum_{j\in\H\backslash\H^t}g_j^t} \nonumber \\
    &=\iprod{x^t-x_{\H}}{\sum_{j\in\H}g_j^t}-\iprod{x^t-x_{\H}}{\sum_{j\in\H\backslash\H^t}g_j^t}.
\end{align*}
Recall that $g_j^t=\nabla Q_j(x^t)$, $\forall j\in\H$. Therefore,
\begin{align}
    \iprod{x^t-x_{\H}}{\sum_{j\in\H^t}g_j^t}=&\iprod{x^t-x_{\H}}{\sum_{j\in\H}\nabla Q_j(x^t)} -\iprod{x^t-x_{\H}}{\sum_{j\in\H\backslash\H^t}\nabla Q_j(x^t)}. \label{eqn:first_phi_1-alt}
\end{align}
~

\noindent Due to the strong convexity assumption (i.e., Assumption~\ref{assum:strongly-convex}), for all $x, \, y \in \W$,
\begin{equation}
    \iprod{x- y}{\nabla \sum_{j\in\H}Q_j(x)-\nabla\sum_{j\in\H} Q_j(y)} \geq \mnorm{\H}\, \gamma\norm{x-y}^2.
    \label{eqn:assum-strong-convex-restate-alt}
\end{equation}
By the projection on $\W$ in update rule \eqref{eqn:update}, $x^t\in\W$ for all $t$. By \eqref{eqn:x_H_W}, $x_{\H}$ is also in $\W$. As $x_{\H}$ is minimum point of $\sum_{j \in \H}Q_j(x)$, $\nabla \sum_{j \in \H} Q_j(x_{\H}) = 0$. Thus, by \eqref{eqn:assum-strong-convex-restate-alt} it follows that
\begin{align}
    \iprod{x^t-x_{\H}}{\sum_{j\in\H}\nabla Q_j(x^t)} & = \iprod{x^t-x_{\H}}{\nabla\sum_{j\in\H}Q_j(x^t)-\nabla\sum_{j\in\H}Q_j(x_{\H})} \geq \mnorm{\H} \, \gamma\norm{x^t-x_{\H}}^2.
    \label{eqn:inner-prod-h-alt}
\end{align}
Now, due to the Cauchy-Schwartz inequality and \eqref{eqn:lipschitz-distance-S1-complement} in Lemma~\ref{lemma:bounding-gradients-at-xH}, with $\H\backslash\H^t\subset\H$ and $\mnorm{\H\backslash\H^t}=f$, we have
\begin{align}
    \iprod{x^t-x_{\H}}{\sum_{j\in\H\backslash\H^t}\nabla Q_j(x^t)}& \leq\norm{x^t-x_{\H}}\cdot\norm{\sum_{j\in\H\backslash\H^t}\nabla Q_j(x^t)}\leq\norm{x^t-x_{\H}}\cdot(n-2f)\mu\epsilon.
    \label{eqn:inner-prod-h-ht-1-alt}
\end{align}
Substituting from~\eqref{eqn:inner-prod-h-alt} and~\eqref{eqn:inner-prod-h-ht-1-alt} in~\eqref{eqn:first_phi_1-alt} we obtain that
\begin{equation}
    \iprod{x^t-x_{\H}}{\sum_{j\in\H^t}g_j^t} \geq \gamma\mnorm{\H}\, \norm{x^t-x_{\H}}^2-(n-2f)\mu\epsilon\norm{x^t-x_{\H}}.
    \label{eqn:phi-t-first-part-alt}
\end{equation}
~

Next, we consider the second term on the right-hand side of~\eqref{eqn:phi-t-two-parts-alt}. From the Cauchy-Schwartz inequality, 
\begin{equation*}
    \iprod{x^t-x_{\H}}{g_k^t}\geq-\norm{x^t-x_{\H}}\, \norm{g_k^t}.
\end{equation*}
Substituting from~\eqref{eqn:phi-t-first-part-alt} and above in~\eqref{eqn:phi-t-two-parts-alt} we obtain that
\begin{equation}
    \phi_t\geq\gamma\mnorm{\H}\, \norm{x^t-x_{\H}}^2-(n-2f)\mu\epsilon\norm{x^t-x_{\H}}-\sum_{k\in\B^t}\norm{x^t-x_{\H}}\, \norm{g_k^t}.
    \label{eqn:phi-t-2-alt}
\end{equation}

Recall the definition of $\B^t$ and $\H^t$. We have $\mnorm{\B^t} = n-f-\mnorm{\H^t}=f$, and $\mnorm{\H\backslash\H^t} = n- f - \mnorm{\H^t}=f$. Also, due to the sorting of the gradients, for each $k \in \B^t$, there are two cases: 
\begin{enumerate}[label=(\roman*)]
    \item $k\in\B^t\cap(\H\backslash\H^t)$, therefore
    \begin{align}
        \norm{g_k^t} = \norm{\nabla Q_k(x^t)}.
    \end{align}
    \item $k\in \B^t\backslash(\H\backslash\H^t)$. Note that $\mnorm{\B^t\backslash(\H\backslash\H^t)}=\mnorm{(\H\backslash\H^t)\backslash\B^t}$. Therefore, there exists a distinct $j\in(\H\backslash\H^t)\backslash\B^t$ for each $k$ in case (ii) that 
    \begin{equation}
        \norm{g_k^t}\leq\norm{g_j^t}=\norm{\nabla Q_j(x^t)}. \label{eqn:k_B_j_H-0-alt}
    \end{equation}
\end{enumerate}
Combining the two cases, each $k\in\B^t$ can always be mapped to a different $j\in\H\backslash\H^t$, such that
\begin{equation}
    \norm{g_k^t}\leq\norm{g_j^t}=\norm{\nabla Q_j(x^t)}. \label{eqn:k_B_j_H-alt}
\end{equation}
Therefore,~\eqref{eqn:k_B_j_H-alt} implies that
\begin{align*}
    \sum_{k \in \B^t} \norm{g_k^t} \leq \sum_{j \in \H\backslash\H^t} \norm{\nabla Q_j(x^t)}.
\end{align*}
Substituting from above in~\eqref{eqn:phi-t-2-alt}, we obtain that
\begin{align*}
    \phi_t&\geq\gamma\mnorm{\H}\, \norm{x^t-x_{\H}}^2-(n-2f)\mu\epsilon\norm{x^t-x_{\H}}-\sum_{j\in\H\backslash\H^t}\norm{x^t-x_{\H}}\, \norm{\nabla Q_j(x^t)}.
\end{align*}
Substituting from the first half of \eqref{eqn:honest-bound-everywhere}, i.e., $\norm{\nabla Q_i(x)}\leq2(n-2f)\mu\epsilon+\mu\norm{x-x_{\H}}$, above, we obtain that
\begin{align*}
    \phi_t&\geq\gamma\mnorm{\H}\, \norm{x^t-x_{\H}}^2 - (n-2f)\mu\epsilon\norm{x^t-x_{\H}} - \mnorm{\H\backslash\H^t}\, \norm{x^t-x_{\H}}\, (2(n-2f)\mu\epsilon+\mu\norm{x^t-x_{\H}}) \nonumber\\
    & \geq \left(\gamma\mnorm{\H}-\mu\mnorm{\H\backslash\H^t}\right)\norm{x^t-x_{\H}}^2 - (n-2f)\mu\epsilon\norm{x^t-x_{\H}}- 2(n-2f)\mu\epsilon \mnorm{\H\backslash\H^t}\, \norm{x^t-x_{\H}}.
\end{align*}
As $\mnorm{\H} = n-f$ and $\mnorm{\H\backslash\H^t}= f$, the above implies that
\begin{align}
\begin{split}
\label{eqn:phi-t-pre-final-alt}
    \phi_t & \geq\left(\gamma(n-f)-\mu f\right)\norm{x^t-x_{\H}}^2-(n-2f)\mu\epsilon\norm{x^t-x_{\H}}-2(n-2f)\mu\epsilon f\, \norm{x^t-x_{\H}} \\
        & = \left(\gamma(n-f)-\mu f\right)\norm{x^t-x_{\H}}\, \left(\norm{x^t-x_{\H}}-\dfrac{(1+2f)(n-2f)\mu\epsilon}{\gamma(n-f)-\mu f} \right) \\
        & = n \gamma \, \left( 1 - \frac{f}{n} \left( 1 + \frac{\mu}{\gamma}\right)\right) \norm{x^t-x_{\H}} \, \left(\norm{x^t-x_{\H}} - \frac{(1+2f)(n-2f)\mu\epsilon}{n\gamma\left( 1 - \frac{f}{n}\left(1 + \frac{ \mu}{\gamma} \right)\right)} \right).
\end{split}
\end{align}
Recall from~\eqref{eqn:alpha-alt} that
$$\alpha = 1-\dfrac{f}{n}\left(1+\dfrac{\mu}{\gamma}\right). $$
Substituting from above in~\eqref{eqn:phi-t-pre-final-alt} we obtain that
\begin{align}
    \phi_t \geq \alpha \, n \gamma \,\norm{x^t-x_{\H}}\, \left(\norm{x^t-x_{\H}}- \frac{(1+2f)(n-2f)\mu}{\alpha n\gamma} \epsilon \right).
    \label{eqn:phi-t-final-alt}
\end{align}
As it is assumed that $\alpha > 0$,~\eqref{eqn:phi-t-final-alt} implies that for an arbitrary $\delta>0$,
\begin{align*}
    \phi_t \geq \alpha n \gamma \delta\left(\cfrac{(1+2f)(n-2f)\mu}{\alpha n\gamma}\epsilon+\delta\right) > 0 ~ \text{ when } \norm{x^t-x_{\H}} \geq \cfrac{(1+2f)(n-2f)\mu}{\alpha n\gamma}\epsilon+\delta.
\end{align*}
Hence, the proof.

\section{Appendix: Proof of Theorem~\ref{thm:conv-guarantee_cwtm}}
\label{app:proof_2}

In this section we present the proof of Theorem~\ref{thm:conv-guarantee_cwtm}.
Throughout this section we assume $f > 0$ to ignore the trivial case of $f = 0$. The proof closely follows that of Theorem~\ref{thm:conv-guarantee}, and we may borrow some notation and results directly from Appendix~\ref{app:proof_1}.\\

Recall from~\eqref{eqn:cwtm_gf} that for CWTM gradient-filter, for all $l \in \{1, \ldots, \, d\}$,
\begin{align}
    \gf\left(g^t_1, \ldots, \, g^t_n \right)[l] = \frac{1}{n-2f} \, \sum_{j=f+1}^{n-f}g_{i_j[l]}^t [l], \quad \forall t. \label{eqn:gf_cwtm_proof}
\end{align}
~

First, we show that $\sum_{j=f+1}^{n-f}g_{i_j[l]}^t [l]$ is finite for all $l$ and $t$. From~\eqref{eqn:honest-bound-everywhere} in Appendix~\ref{app:proof_1}, we know that if the $(2f, \, \epsilon)$-redundancy property and Assumption~\ref{assum:lipschitz} hold true then for each non-faulty agent $i\in\H$, 
\begin{equation}
    \norm{\nabla Q_i(x)} \leq 2n\mu\epsilon+\mu\norm{x-x_{\H}}. \label{eqn:pf_2_lip}
\end{equation}
The above implies that for all $i \in \H$, $l \in \{1, \ldots, \, d\}$ and $x$,
\begin{equation}
    \mnorm{\nabla Q_i(x)[l]} \leq 2n\mu\epsilon+\mu\norm{x-x_{\H}}.
    \label{eqn:honest-bound-everywhere_l}
\end{equation}
Recall that for all $l$ and $t$, 
\begin{equation*}
    g_{i_1[l]}^t[l] \leq \ldots \leq g_{i_{f + 1}[l]}^t[l] \leq \ldots \leq g_{i_{n-f}[l]}^t[l] \leq  \ldots \leq g_{i_n[l]}^t[l].
\end{equation*}
As there are at most $f$ Byzantine agents and $\mnorm{\H} = n-f$, for all $l$ and $t$ there exists a pair of non-faulty agents $\sigma^1_t[l], \,  \sigma^2_t[l] \in \H$ such that
\begin{equation}
    g_{i_{n-f}[l]}^t[l] \leq g_{i_{\sigma^1_t[l]}}^t[l], \text{ and } g_{i_{f+1}[l]}^t[l] \geq g_{i_{\sigma^2_t[l]}}^t[l].
    \label{eqn:honest-bound_l}
\end{equation}
As $g_j^t=\nabla Q_j(x^t)$ for all $j\in\H$, from~\eqref{eqn:honest-bound_l} we obtain that for all $j \in \{f+1, \ldots, n-f\}$, $l$ and $t$,
\begin{equation*}
    \mnorm{g_{i_j[l]}^t[l]} \leq \max \left\{ \mnorm{\nabla Q_{\sigma^1_t[l]}(x^t)[l]}, \, \mnorm{\nabla Q_{\sigma^1_t[l]}(x^t)[l]} \right\}.
\end{equation*}
Substituting from~\eqref{eqn:honest-bound-everywhere_l} above we obtain that for all $j \in \{f+1, \ldots, n-f\}$, $l$ and $t$,
\begin{equation*}
    \mnorm{g_{i_j[l]}^t[l]} \leq  2n\mu\epsilon+\mu\norm{x^t-x_{\H}}.
\end{equation*}
Therefore, owing to the triangle inequality,
\begin{equation}
    \mnorm{\sum_{j=f+1}^{n-f}g_{i_j[l]}^t [l]} \leq \sum_{j=f+1}^{n-f}\mnorm{g_{i_j[l]}^t [l]} \leq (n-2f)\left(2n\mu\epsilon+\mu\norm{x^t-x_{\H}}\right). \label{eqn:filtered-upperbound_l}
\end{equation}
Let $\Gamma = \max_{x \in \W} \norm{x - x_{\H}}$. As $\W$ is a compact set, $\Gamma < \infty$. Recall from the update rule~\eqref{eqn:update} that $x^t \in \W$ for all $t$. Thus, $\norm{x^t - x_{\H}} \leq \max_{x \in \W} \norm{x - x_{\H}} = \Gamma < \infty$. Substituting this in~\eqref{eqn:filtered-upperbound_l} implies that for all $l \in \{1, \ldots, \, d\}$,
\begin{equation}
    \mnorm{\sum_{j=f+1}^{n-f}g_{i_j[l]}^t [l]} \leq (n-2f) \left( 2n \mu \epsilon + \mu \Gamma\right) < \infty. 
    \label{eqn:cwtm_bnd_per_coordinate}
\end{equation}
Combining \eqref{eqn:cwtm_bnd_per_coordinate} with~\eqref{eqn:gf_cwtm_proof}, by the definition of Euclidean norm, we obtain that 
\begin{align}
    \norm{\gf \left(g^t_1, \ldots, \, g^t_n \right)}&= \left(\mnorm{{\gf \left(g^t_1, \ldots, \, g^t_n \right)[1]}}^2 + ... + \mnorm{\gf \left(g^t_1, \ldots, \, g^t_n \right)[d]^2}\right)^{1/2} \nonumber \\
    &=\left(\mnorm{\frac{1}{n-2f}\sum_{j=f+1}^{n-f}g_{i_j[1]}^t [1]}^2 + ... + \mnorm{\frac{1}{n-2f}\sum_{j=f+1}^{n-f}g_{i_j[d]}^t [d]}^2\right)^{1/2} \nonumber \\
    &=\left(\frac{1}{(n-2f)^2}\mnorm{\sum_{j=f+1}^{n-f}g_{i_j[1]}^t [1]}^2 + ... + \frac{1}{(n-2f)^2}\mnorm{\sum_{j=f+1}^{n-f}g_{i_j[d]}^t [d]}^2\right)^{1/2} \nonumber \\
    &\leq\left(\underbrace{\left( 2n \mu \epsilon + \mu \Gamma\right)^2+...+\left( 2n \mu \epsilon + \mu \Gamma\right)^2}_{d\textrm{ times}}\right)^{1/2}< \infty, \quad \forall t. \label{eqn:cwtm_bnd_grd}
\end{align}
~

Now, consider an arbitrary iteration $t$ and $l \in \{1, \ldots, \, d\}$. From prior works on CWTM gradient-filter for the scalar case~\cite{su2016fault}, i.e., when $d = 1$, we know that trimmed mean of the $l$-th elements of the gradients lies in the convex hull of $l$-th elements of the non-faulty agents' gradients in set $\H$. Specifically,
\begin{align}
    \min_{i \in \H}g^t_{i}[l] \leq \gf\left(g^t_1, \ldots, \, g^t_n \right)[l] \leq \max_{i \in \H} g^t_{i}[l]. \label{eqn:cwtm_bnd_1}
\end{align}
Obviously,
\begin{align}
    \min_{i \in \H}g^t_{i}[l] \leq  \frac{1}{\mnorm{\H}} \sum_{i \in \H} g^t_{i}[l] \leq \max_{i \in \H} g^t_{i}[l]. \label{eqn:cwtm_bnd_2}
\end{align}
Therefore, from~\eqref{eqn:cwtm_bnd_1} and~\eqref{eqn:cwtm_bnd_2} we obtain that
\begin{align*}
    \mnorm{\gf\left(g^t_1, \ldots, \, g^t_n \right)[l] - \frac{1}{\mnorm{\H}} \sum_{i \in \H} g^t_{i}[l]} \leq \max_{i \in \H} g^t_{i}[l] - \min_{i \in \H} g^t_i[l]. 
\end{align*}
As $\max_{i \in \H} g^t_{i}[l] - \min_{i \in \H} g^t_i[l] = \max_{i, \, j \in \H} \mnorm{g^t_{i}[l] - g^t_{j}[l]}$, and $g^t_i = \nabla Q_i(x^t)$ for all $i \in \H$, the above can be re-written as follows.
\begin{align}
    & \mnorm{\gf\left(g^t_1, \ldots, \, g^t_n \right)[l] - \frac{1}{\mnorm{\H}} \sum_{i \in \H} \nabla Q_i(x^t)[l]} 
        \leq \max_{i, \, j \in \H} \mnorm{\nabla Q_i(x^t)[l] - \nabla Q_j(x^t)[l]}.  \label{eqn:cwtm-avg}
\end{align}
Note that for any two $i, \, j \in \H$,
\begin{align}
    \mnorm{\nabla Q_i(x^t)[l] - \nabla Q_j(x^t)[l]} \leq \norm{\nabla Q_i(x^t) - \nabla Q_j(x^t)}. \label{eqn:sep_bnd_1}
\end{align}
Substituting from Assumption~\ref{assum:corr_grad}, 
\[\norm{\nabla Q_i(x) - \nabla Q_j(x)} \leq \lambda \max \left\{\norm{\nabla Q_i(x)}, \norm{\nabla Q_j(x)} \right\},\] in~\eqref{eqn:sep_bnd_1} we obtain that
\begin{align}
    \mnorm{\nabla Q_i(x^t)[l] - \nabla Q_j(x^t)[l]} \leq \lambda \max \left\{\norm{\nabla Q_i(x^t)}, \, \norm{\nabla Q_j(x^t)} \right\}. \label{eqn:sep_bnd_2}
\end{align}
Substituting from~\eqref{eqn:pf_2_lip} above we obtain that
\begin{align}
    \mnorm{\nabla Q_i(x^t)[l] - \nabla Q_j(x^t)[l]} \leq \lambda \left( 2n\mu \, \epsilon+\mu\norm{x^t-x_{\H}} \right). \label{eqn:sep_bnd_3}
\end{align}
Finally, substituting from~\eqref{eqn:sep_bnd_3} in~\eqref{eqn:cwtm-avg} we obtain that, for all $l$,
\begin{align*}
    & \mnorm{\gf\left(g^t_1, \ldots, \, g^t_n \right)[l] - \frac{1}{\mnorm{\H}} \sum_{i \in \H} \nabla Q_i(x^t)[l]} 
    \leq \lambda \left( 2n\mu \, \epsilon+\mu\norm{x^t-x_{\H}} \right). 
\end{align*}
As $\norm{x} = \sqrt{\sum_{l = 1}^d \mnorm{x[l]}^2}$ for $x \in \R^d$, the above implies that
\begin{align}
    & \norm{\gf\left(g^t_1, \ldots, \, g^t_n \right) - \frac{1}{\mnorm{\H}} \sum_{i \in \H} \nabla Q_i(x^t)}  \leq \sqrt{d} \lambda \left( 2n\mu \, \epsilon+\mu\norm{x^t-x_{\H}} \right). \label{eqn:cwtm-avg-2}
\end{align}
Now, note that
\begin{align}
     & \gf\left(g^t_1, \ldots, \, g^t_n \right) = \frac{1}{\mnorm{\H}} \sum_{i \in \H} \nabla Q_i(x^t)
     + \left( \gf\left(g^t_1, \ldots, \, g^t_n \right) - \frac{1}{\mnorm{\H}} \sum_{i \in \H} \nabla Q_i(x^t) \right).  \label{eqn:cwtm-cwtm-avg}
\end{align}
Recall from Theorem~\ref{thm:upper-bound-D} that $\phi_t$, for each $t$, is defined to be
\begin{align*}
    \phi_t = \iprod{x^t-x_{\H}}{\gf\left(g^t_1, \ldots, \, g^t_n \right)}.
\end{align*}
Substituting from~\eqref{eqn:cwtm-cwtm-avg} above we obtain that
\begin{align}
    \phi_t & = \iprod{x^t - x_{\H}}{\frac{1}{\mnorm{\H}} \sum_{i \in \H} \nabla Q_i(x^t)} +  \iprod{x^t - x_{\H}}{\gf\left(g^t_1, \ldots, \, g^t_n \right) - \frac{1}{\mnorm{\H}} \sum_{i \in \H} \nabla Q_i(x^t)}. \label{eqn:phi_t_cwtm_1}
\end{align}
~

Recall from Assumption~\ref{assum:strongly-convex} that $Q_{\H}(x) = (1/\mnorm{\H})\sum_{i \in \H}Q_i(x)$. Thus, the first term on the right-hand side of~\eqref{eqn:phi_t_cwtm_1}, 
\begin{align*}
    \iprod{x^t - x_{\H}}{\frac{1}{\mnorm{\H}} \sum_{i \in \H} \nabla Q_i(x^t)} = \iprod{x^t - x_{\H}}{\nabla Q_{\H}(x^t)}.
\end{align*}
Substituting from the Assumption~\ref{assum:strongly-convex} above, and recalling that $\nabla Q_{\H}(x_{\H}) = 0$, we obtain that
\begin{align}
    \iprod{x^t - x_{\H}}{\frac{1}{\mnorm{\H}} \sum_{i \in \H} \nabla Q_i(x^t)}  \geq \gamma \, \norm{x^t - x_{\H}}^2. \label{eqn:first_phi_cwtm}
\end{align}
Next, we consider the second term on the right-hand side of~\eqref{eqn:phi_t_cwtm_1}. From Cauchy-Schwartz inequality,
\begin{align*}
    & \iprod{x^t - x_{\H}}{\gf\left(g^t_1, \ldots, \, g^t_n \right) - \frac{1}{\mnorm{\H}} \sum_{i \in \H} \nabla Q_i(x^t)}  \nonumber \\
    \geq& - \norm{x^t - x_{\H}} \norm{\gf\left(g^t_1, \ldots, \, g^t_n \right) - \frac{1}{\mnorm{\H}} \sum_{i \in \H} \nabla Q_i(x^t)}.
\end{align*}
Substituting from~\eqref{eqn:cwtm-avg-2} above we obtain that
\begin{align}
    & \iprod{x^t - x_{\H}}{\gf\left(g^t_1, \ldots, \, g^t_n \right) - \frac{1}{\mnorm{\H}} \sum_{i \in \H} \nabla Q_i(x^t)} \nonumber \\
    & \geq - \sqrt{d} \lambda \, \norm{x^t - x_{\H}}  \left( 2n\mu \, \epsilon+\mu\norm{x^t-x_{\H}} \right). \label{eqn:second_phi_cwtm}
\end{align}
Substituting from~\eqref{eqn:first_phi_cwtm} and~\eqref{eqn:second_phi_cwtm} in~\eqref{eqn:phi_t_cwtm_1} we obtain that
\begin{align}
\begin{split}
    \phi_t & \geq \gamma \, \norm{x^t - x_{\H}}^2 - \sqrt{d} \lambda \, \norm{x^t - x_{\H}}  \left( 2n\mu \, \epsilon+\mu\norm{x^t-x_{\H}} \right) \label{eqn:phi_t_cwtm_2} \\
    & = \left(\gamma -  \sqrt{d} \lambda \mu \right) \norm{x^t - x_{\H}} \left( \norm{x^t - x_{\H}} - \frac{2 \sqrt{d} n \mu \lambda}{(\gamma - \sqrt{d} \mu \lambda)} \, \epsilon \right).
\end{split}
\end{align}
Therefore, if, for an arbitrary $\delta > 0$, 
\begin{align*}
    \norm{x^t - x_{\H}} \geq \frac{2 \sqrt{d} n \mu \lambda}{(\gamma - \sqrt{d} \mu \lambda)} \, \epsilon  + \delta
\end{align*}
then
\begin{align*}
    \phi_t &\geq \left(\gamma -  \sqrt{d} \lambda \mu \right) \left( \frac{2 \sqrt{d} n \mu \lambda}{(\gamma - \sqrt{d} \mu \lambda)} \, \epsilon  + \delta \right) \, \delta = \left( 2 \sqrt{d} n \mu \lambda \, \epsilon +  \left(\gamma -  \sqrt{d} \lambda \mu \right) \, \delta \right) \, \delta.
\end{align*}
Hence, the proof.
\section{Details on the Numerical Experiments}
\label{apdx:experiments}

In this section, we present the details of the simulation results to empirically compare the approximate fault-tolerance achieved by the aforementioned gradient-filters; CGE and CWTM, extending Section~\ref{sec:experiments}. For the simulation, we consider the problem of distributed linear regression, which is a special distributed optimization problem with quadratic cost functions~\cite{gupta2019byzantine}. 

\subsection{Problem description}
We consider a synchronous server-based system, as shown in Figure~\ref{fig:sys}, wherein $n=6$, $d=2$, and $f=1$. Each agent $i\in \{1, \, \ldots, \, n\}$ has a data point represented by a triplet $(A_i, \, B_i, \, N_i)$ where $A_i$ is a $d$-dimensional row vector, $B_i \in \R$ as the response, and a noise value $N_i \in \R$. Specifically, for all $i \in \{1, \, \ldots, \, n\}$, 
\begin{equation}
    B_i=A_ix^*+N_i ~ \text{ where } ~x^*=\begin{pmatrix}1\\1\end{pmatrix}.
\end{equation}
The collective data is represented by a triplet of matrices $(A, \, B, \, N)$ where the $i$-th row of $A$, $B$, and $N$ are equal to $A_i$, $B_i$ and $N_i$, respectively. The specific values are as follows.
\begin{equation}
    A=\begin{pmatrix}
        1 & 0 \\
        0.8 & 0.5 \\
        0.5 & 0.8 \\
        0 & 1 \\
        -0.5 & 0.8 \\
        -0.8 & 0.5 
    \end{pmatrix}, ~ 
    B=\begin{pmatrix}
        0.9108 \\  1.3349 \\ 1.3376 \\ 1.0033 \\ 0.2142 \\ -0.3615
    \end{pmatrix}, ~ \text{ and }
    N=\begin{pmatrix}
        -0.0892 \\0.0349 \\0.0376\\0.0033\\-0.0858\\-0.0615
    \end{pmatrix}.
\end{equation}
It should be noted that
\begin{align}
    B = A x^{*} + N. \label{eqn:set_equations}
\end{align}
We let $A_S$, $B_S$ and $N_S$ represent matrices of dimensions $\mnorm{S} \times 2$, $\mnorm{S} \times 1$ and $\mnorm{S} \times 1$ obtained by stacking the rows $\{A_i, \, i \in S \}$, $\{B_i, \, i \in S \}$ and $\{N_i, \, i \in S\}$, respectively, in the increasing order. From~\eqref{eqn:set_equations}, observe that for every non-empty set $S$,
\begin{align}
    B_S = A_S x^{*} + N_S. \label{eqn:set_equations-2}
\end{align}
Recall from basic linear algebra that if $A_S$ is full-column rank, i.e., $\rank{A_S} = d = 2$ then $x^*$ is the unique solution of the set of equations in~\eqref{eqn:set_equations-2}. Note that for every set $S$ with $\mnorm{S} \geq n-2f = 6 - 2 = 4$, the matrix $A_S$ is full rank. Specifically,
\begin{align}
    \rank{A_S} = d = 2, \quad \forall S \subseteq \{1, \ldots, \, 6\}, ~ \mnorm{S} \geq 4. \label{eqn:exp_full_rank}
\end{align}
~

In this particular distributed optimization problem, each agent $i$ has a quadratic cost function defined to be 
\[Q_i(x)=(B_i-A_ix)^2, \quad \forall x \in \R^2.\] 
For an arbitrary non-empty set of agents $S$, we define 
\begin{align}
Q_S(x) = \sum_{i\in S}Q_i(x) = \sum_{i \in S}\left(B_i-A_ix \right)^2 = \norm{B_S-A_Sx}^2, \quad \forall x \in \R^2. \label{def:exp_Q_S}
\end{align}
As matrix $A_S$ is full rank for every $S$ with $\mnorm{S} \geq 4$, to minimize $Q_S(x)$ we have \cite{rencher2002methods}
\begin{align}
    \arg \min_{x \in \R^2} Q_S(x) = \arg \min_{x \in \R^2} \norm{B_S-A_Sx}^2 = \left(A_S^TA_S\right)^{-1}A_S^TB_S, \label{eqn:solve_Q_S}
\end{align}
where $A_T$ is the transpose of the matrix $A$.
Therefore, $Q_S(x)$ has a unique minimum point when $\mnorm{S} \geq 4$. Henceforth, we write notation $\arg \min_{x \in \R^2}$ simply as $\arg \min$, unless otherwise stated.

\subsection{Simulations}
\label{sub:simulations}
\comment{(113) and similar places below are updated. The calculations are (weirdly) correct, but only the descriptions here were wrong.\\Figures 2 and 3 are updated.\\TODO: update descriptions and interpretations of Fig. 2 and 3}
Due to the rank condition~\eqref{eqn:exp_full_rank}, the agents' cost functions satisfy the {\em $(2f, \, \epsilon)$-redundancy} property, stated in Definition~\ref{def:approx_red}, with $\epsilon = 0.0890$. The steps for computing $\epsilon$ are described below.
\begin{enumerate}
\setlength{\itemsep}{0.3em}
    \item For each set $S \subset\{1, \ldots, \, 6\}$ with $\mnorm{S} = n-f = 5$, compute $x_S=\left(A_S^TA_S\right)^{-1}A_S^TB_S$. Note that by~\eqref{eqn:solve_Q_S}, $x_S = \arg \min Q_S(x)$.
    
    \item For each set $S \subset\{1, \ldots, \, 6\}$ with $\mnorm{S} = n-f = 5$ do the following: 
    
    \begin{enumerate}
        \item For each set $\widehat{S} \subseteq S$ with $\mnorm{\widehat{S}} \geq n-2f = 4$, compute $x_{\widehat{S}}=\left(A_{\widehat{S}}^TA_{\widehat{S}}\right)^{-1}A_{\widehat{S}}^TB_{\widehat{S}}$. Note that by~\eqref{eqn:solve_Q_S}, $x_{\widehat{S}} = \arg \min Q_{\widehat{S}}(x)$.
        
        \item Compute 
        \[\epsilon_S=\max\limits_{\widehat{S}\subseteq S,\,\mnorm{\widehat{S}}\geq 4} \norm{x_S - x_{\widehat{S}}}.\]
        In this particular case, both the sets of minimum points $\arg \min Q_S(x)$ and $\arg \min Q_{\widehat{S}}(x)$ are singleton with points $x_S$ and $x_{\widehat{S}}$, respectively. Therefore, 
        \[\norm{x_S - x_{\widehat{S}}} = \dist{\arg \min Q_S(x)}{ \arg \min Q_{\widehat{S}}(x)}.\]
    \end{enumerate}
    
    \item In the final step, we compute 
    \[\epsilon=\max\limits_{\mnorm{S}=n-f}\epsilon_S.\]
\end{enumerate}
~

For each agent $i$, its cost function $Q_i(x)$ has Lipschitz continuous gradients, i.e., satisfy Assumption~\ref{assum:lipschitz}, with Lipschitz coefficient
\begin{align}
    \mu = \overline{v}_i \label{eqn:exp_lip}
\end{align}
where $\overline{v}_i$ denotes the largest eigenvalue of $A_i^TA_i$. Also, for every set of agents $S$ with $\mnorm{S} = n-f = 5$, their average cost function $(1/\mnorm{S}) Q_S(x)$ is strongly convex, i.e., satisfy Assumption~\ref{assum:strongly-convex}, with the strong convexity coefficient
\begin{align}
    \gamma = \frac{1}{\mnorm{S}}\underline{v}_S \label{eqn:exp_str}
\end{align}
where $\underline{v}_S$ is the smallest eigenvalue of $A_S^TA_S$. Derivations of~\eqref{eqn:exp_lip} and~\eqref{eqn:exp_str} can found in \cite[Section 10]{gupta2019byzantine}.


\subsection{Simulation}
In our experiments, we simulate the following fault behaviors for the faulty agent.
\begin{itemize}
\setlength{\itemsep}{0.3em}
    \item \textit{gradient-reverse}: the faulty agent \textit{reverses} its true gradient. Suppose the correct gradient of a faulty agent $i$ at iteration $t$ is $s_i^t$, the agent $i$ will send the incorrect gradient $g_i^t=-s_i^t$ to the server.
    \item \textit{random}: the faulty agent sends a randomly chosen vector in $\mathbb{R}^d$. In our experiments, the faulty agent in each iteration chooses i.i.d. Gaussian random vector with mean 0 and a isotropic covariance matrix with standard deviation of 200.
\end{itemize}

We simulate the distributed gradient-descent algorithm described in Section~\ref{sub:steps} by assuming agent $1$ to be Byzantine faulty. It should be noted that the identity of the faulty agent is not used in any way during the simulations. Here, the set of non-faulty agents is $\H = \{2, \ldots, \, 6\}$ and $\mnorm{\H} = n-f = 5$. Therefore, in this particular case $\H$ is the only set of $n-f$ non-faulty agents. From~\eqref{eqn:solve_Q_S}, we obtain that the minimum point of the aggregate cost function $\sum_{i \in \H}Q_i(x)$, denoted by $x_{\H}$, is equal to the solution of the following set of linear equations: 
$$B_\H=A_\H x_\H.$$
Specifically, $x_\H=\begin{pmatrix}1.0780\\ 0.9825\end{pmatrix}$. Also, note from our earlier deductions in~\eqref{eqn:exp_lip} and~\eqref{eqn:exp_str} that in this particular case, the non-faulty agents' cost functions satisfy Assumptions~\ref{assum:lipschitz} and~\ref{assum:strongly-convex} with $\mu=1$ and $\gamma=0.356$, respectively.\\

{\bf Parameters:} We use the following parameters for implementing the algorithm. In the update rule~\eqref{eqn:update}, we use step-size $\eta_t=1.5/(t+1)$ for iteration $t = 0, \, 1, \ldots$. Note that this particular step-size is diminishing and satisfies the conditions: $\sum_{t=0}^\infty\eta_t=\infty$ and $\sum_{t=0}^\infty\eta_t^2=3\pi^2/8<\infty$ (see~\cite{rudin1964principles}). We assume the convex compact $\W \subset \R^d$ to be a 2-dimensional hypercube $[-1000,1000]^2$. Note that $x_{\H} \in \W$, i.e., Assumption~\ref{assum:compact} holds true. In all the simulation results presented below, the initial estimate $x^0 = (0,0)^T$. \\

In every execution, we observe that the iterative estimates produced by the algorithm practically converge after $400$ iterations. Thus, to measure the approximate fault-tolerance achieved by the different gradient-filter, i.e., CGE and CWTM, we define the output of the algorithm to be $x_{\textrm{out}} = x^{500}$. The outputs for the two gradient-filters, under different faulty behaviors, are shown in Table~\ref{tab:results}. Note that $\dist{x_\H}{x_{\mathrm{out}}}=\norm{x_\H-x_{\mathrm{out}}}$. The results for the case when the faulty agent sends {\em random} faulty gradients are only shown for a randomly chosen execution. \\

{\bf Conclusion:} As shown in Table~\ref{tab:results}, in all executions, the distances between $x_\H$ the output of the algorithm $x_{\textrm{out}}$ in case of both CGE and CWTM gradient-filters are smaller than $\epsilon$. For the said executions, we plot in Figure~\ref{fig:fault-comparison} the values of the aggregate cost function $\sum_{i \in \H}Q_i(x^t)$ (referred as {\em loss}) and the approximation error $\norm{x^t - x_{\H}}$ (referred as {\em distance}) for iteration $t$ ranging from $0$ to $1500$. We also show the plots of the fault-free distributed gradient-descent (DGD) method where the faulty agent is omitted, 
and the DGD method without any gradient-filter when agent $1$ is Byzantine faulty. The details for iteration $t$ ranging from 0 to 80 are also highlighted in Figure~\ref{fig:fault-comparison-detail}. 
\section{Empirical results on machine learning tasks}
\label{appdx:exp-learning}

In this section, we present some experiment results applying our algorithm with DGD proposed in Section~\ref{sub:steps} on distributed machine learning tasks where in the system there are Byzantine faulty agents. The applicability of our algorithm with to distributed learning problems is discussed in Section~\ref{sub:app}. Furthermore, our algorithm can be adapted to use distributed stochastic gradient descent (D-SGD) instead of DGD. Suppose each agent $i$ samples $b$ data points in each iteration $t$ from its local data generating distribution $\mathcal{D}_i$, and let those data points be $\boldsymbol{z}_i^t=\{z_{i_1}^t,...,z_{i_b}^t\}$. Each agent computes its stochastic gradient $g_i^t$ by
\[g_i^t=\frac{1}{b}\sum_{z\in\boldsymbol{z}_i^t}\nabla \ell(x^t; z),\]
and sends it to the server. Thus, $g_i^t$ is expected to be the gradient of cost function $Q_i(x)$ at $x^t$, since
\[\mathbb{E}_{\boldsymbol{z}_i^t}\frac{1}{b}\sum_{z\in\boldsymbol{z}_i^t}\ell (x; z)=\frac{1}{b}\sum_{z\in\boldsymbol{z}_i^t}\mathbb{E}_{z\sim\mathcal{D}_i} \ell (x; z)=\frac{1}{b}\sum_{z\in\boldsymbol{z}_i^t}Q_i(x)=Q_i(x).\]
Therefore, the correctness of our algorithm holds in expectation. \\

It is worth noting that to distributed machine learning problems, it is difficult to compute the exact approximation parameters $\epsilon$, or to ensure the cost functions hold the convexity and smoothness assumptions (i.e., to compute the value of $\mu$ and $\gamma$), as we did in Section~\ref{sec:experiments} or Appendix~\ref{apdx:experiments}. However, the algorithm and our theoretical analysis are still worth considering, since prior research has pointed out that cost functions of many machine learning problems are strongly-convex in the neighborhood of local minimizers \cite{bottou2018optimization}. We conduct these experiments to show empirically the broader applicability of our algorithm, and that the redundancy in cost functions we discussed in this paper indeed exists in real-world scenarios.\\

In our experiments, we simulate a multi-agent distributed learning system in the server-based architecture using multiple threads. There is one thread representing the server, and others each representing one agent. The inter-thread communication is handled via a message passing interface (MPI). The simulator is built in Python using PyTorch \cite{paszke2019pytorch} and MPI4py \cite{dalcin2011parallel}. The simulator is deployed on a Google Cloud Platform cluster with 14 vCPUs and 100 GB memory. \\

We conduct our experiments on two benchmark datasets, MNIST \cite{bottou1998online} and Fashion-MNIST \cite{xiao2017fashion}. Both datasets are monochrome image-classification tasks, and comprise of 60.000 training and 10,000 testing data points, with each image of size $28\times28$. The training and testing data points in each dataset are evenly divided into 10 non-overlapping classes. For each dataset, we train a benchmark neural network Lenet \cite{lecun1998gradient} with 431,080 learnable parameters, i.e., $d=431,080$. \\

In each of our experiments, we simulate a multi-agent system with $n=10$ agents. Out of $n$ agents, we randomly select $f=3$ to be Byzantine faulty. For faulty agents in our experiments, we consider two types of faults, listed as follows:
\begin{itemize}[nosep]
    \item \textbf{Label-flipping} (LF): we scramble the classification labels of data possessed by designated faulty agents. Specifically, suppose a data point of a faulty agent has a label $y\in\{0,...,9\}$, it is changed to $\widetilde{y}=9-y$.
    \item \textbf{Gradient-reverse} (GR): the faulty agents reverse their true gradients. Specifically, suppose the correct gradient of a faulty agent $i$ at iteration $t$ is $s_i^t$, the agent $i$ will send the incorrect gradient $g_i^t=-s_i^t$ to the server, the same as it is defined in Section~\ref{sec:experiments} and Appendix~\ref{apdx:experiments}.
\end{itemize}
We also choose the following parameters: batch size $b=128$, step-size $\eta=0.01$. In each experiment, the training dataset of 60,000 data points are randomly and evenly divided into 10 subsets, one for each agent. \\

\begin{figure}[t]
    \centering
    \includegraphics[width=.75\linewidth]{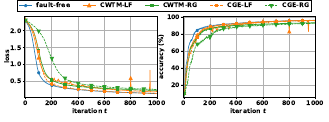}
    \caption{\small{\it The cross-entropy loss and model accuracy, versus the number of iterations in the algorithm, using our algorithm with D-SGD on MNIST with $n=10$ and $f=3$. 
    The two experiments using CWTM are in \em{solid} plots, and the two using CGE are in \em{dashed} plots. The two experiments against LF faults are plotted in \emph{yellow}, while the two against RG are plotted in \emph{green}. 
    We also show the performance of fault-free D-SGD method where the faulty agent is omitted 
    in {\em blue solid} plots.}}
    \label{fig:mnist}
\end{figure}

\begin{figure}[t]
    \centering
    \includegraphics[width=.75 \linewidth]{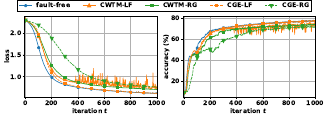}
    \caption{\small{\it The cross-entropy loss and model accuracy, versus the number of iterations in the algorithm, using our algorithm with D-SGD on Fashion-MNIST with $n=10$ and $f=3$. 
    The two experiments using CWTM are in \em{solid} plots, and the two using CGE are in \em{dashed} plots. The two experiments against LF faults are plotted in \emph{yellow}, while the two against RG are plotted in \emph{green}. 
    We also show the performance of fault-free D-SGD method where the faulty agent is omitted 
    in {\em blue solid} plots.}}
    \label{fig:fashion-mnist}
\end{figure}

For each dataset, we compare the two gradient filters discussed in Section~\ref{sec:grad_filters}, CGE and coordinate-wise trimmed mean (CWTM), against two types of faulty agents, LF and GR. To make the results comparable, the random seed is fixed across executions, so that the data points distributed to agents and the selected faulty agents are fixed. The experiments are also compared by a fault-free setting, where the would-have-been faulty agents do not participate the computation. The performance of our algorithm is measured by cross-entropy loss and model accuracy at each step. The results are shown in Figure~\ref{fig:mnist} for MNIST and Figure~\ref{fig:fashion-mnist} for Fashion-MNIST. Since there is a clear trend of converging by the end of 1,000 iterations, we only show performance of the first 1,000 iterations in both figures. \\

As is shown in the two figures, the losses in all our experiments are converging to within a close range of the fault-free results. Still, the performance of the two gradient filters vary when facing different fault types. The differences between performance in presence of Byzantine faulty agents and fault-free case is more significant comparing to numerical experiments in Section~\ref{sec:experiments}, which could be the result of larger value of $\epsilon$ and higher value of $f/n$. Specifically, the performance of CWTM when facing LF faults tabulates during the training process -- which is a possible situation according to our theoretical analysis, since the theoretical results only guarantees that the iterative estimates converges a bounded area around the true minimum. The differences in accuracies between gradient filters and fault-free case are less significant, possibly due to the nature of the loss functions in the two machine learning tasks. Overall, these empirical results indicate that our algorithm can be deployed to solve real-world problem, and the redundancy in cost functions indeed exists in some machine learning problems. 

\end{document}